\documentclass[11pt,letterpaper]{article}

\usepackage[T1]{fontenc}
\usepackage[utf8]{inputenc}
\usepackage{lmodern}
\usepackage[margin=1in]{geometry}
\usepackage{amsthm,amsmath,amsfonts,amssymb, comment}
\usepackage[numbers,sort&compress]{natbib}
\usepackage{microtype}
\usepackage[hidelinks,hypertexnames=false]{hyperref}
\hypersetup{
  pdftitle={Universality of approximate message passing under a fourth-moment tail condition},
  pdfauthor={Ripunjay Dwivedi and Taejoon Kim}
}
\theoremstyle{plain}
\newtheorem{theorem}{Theorem}[section]
\newtheorem{proposition}[theorem]{Proposition}
\newtheorem{lemma}[theorem]{Lemma}

\theoremstyle{definition}
\newtheorem{definition}[theorem]{Definition}
\newtheorem{assumption}[theorem]{Assumption}
\newtheorem{remark}[theorem]{Remark}
\newtheorem{example}[theorem]{Example}

\newcommand{\R}{\mathbb{R}}
\newcommand{\N}{\mathbb{N}}
\newcommand{\E}{\mathbb{E}}
\newcommand{\PP}{\mathbb{P}}
\newcommand{\Var}{\mathrm{Var}}
\newcommand{\bz}{\mathbf{z}}
\newcommand{\bA}{\mathbf{A}}
\newcommand{\bAt}{\widetilde{\mathbf{A}}}
\newcommand{\bAh}{\widehat{\mathbf{A}}}
\newcommand{\bB}{\mathbf{B}}
\newcommand{\bM}{\mathbf{M}}
\newcommand{\bQ}{\mathbf{Q}}
\newcommand{\bX}{\mathbf{X}}
\newcommand{\bY}{\mathbf{Y}}
\newcommand{\bG}{\mathbf{G}}
\newcommand{\bC}{\mathbf{C}}
\newcommand{\bq}{\mathbf{q}}
\newcommand{\bb}{\mathbf{b}}
\providecommand{\bm}{}\renewcommand{\bm}{\mathbf{m}}
\newcommand{\bh}{\mathbf{h}}
\newcommand{\bw}{\mathbf{w}}
\newcommand{\bx}{\mathbf{x}}
\newcommand{\be}{\mathbf{e}}
\newcommand{\bv}{\mathbf{v}}
\newcommand{\bu}{\mathbf{u}}
\newcommand{\by}{\mathbf{y}}
\newcommand{\ba}{\mathbf{a}}
\newcommand{\PM}{\mathbf{P}_{\bM}}
\newcommand{\PMp}{\mathbf{P}_{\bM}^{\perp}}
\newcommand{\PQ}{\mathbf{P}_{\bQ}}
\newcommand{\PQp}{\mathbf{P}_{\bQ}^{\perp}}
\newcommand{\cI}{\mathcal{I}}

\newcommand{\cV}{\mathcal{V}}
\newcommand{\eqd}{\overset{d}{=}}
\DeclareMathOperator{\tr}{tr}
\newcommand{\norm}[1]{\lVert #1 \rVert}
\newcommand{\abs}[1]{\lvert #1 \rvert}
\newcommand{\inner}[2]{\langle #1, #2 \rangle}
\newcommand{\ind}{\mathbf{1}}
\newcommand{\PL}{\mathrm{PL}}

\title{Universality of approximate message passing under a fourth-moment tail condition}
\author{%
  Ripunjay~Dwivedi\textsuperscript{1,2,3}\qquad Taejoon~Kim\textsuperscript{2}\\[0.5em]
  {\small\textsuperscript{1}University of Oxford, Oxford, United Kingdom}\\
  {\small\textsuperscript{2}Arizona State University, Tempe, AZ, USA}\\
  {\small\textsuperscript{3}Indian Institute of Technology Bombay, Mumbai, Maharashtra, India}\\[0.4em]
  {\small\href{mailto:ripunjay.dwivedi@pmb.ox.ac.uk}{\nolinkurl{ripunjay.dwivedi@pmb.ox.ac.uk}}}\\
  {\small\href{mailto:taejoonkim@asu.edu}{\nolinkurl{taejoonkim@asu.edu}}}%
}
\date{}

\begin{document}
\maketitle

\begin{abstract}
We prove state evolution for approximate message passing (AMP) with rectangular random matrices whose entries are independent, centered, and variance-normalized. The row and column dimensions grow with a ratio that converges to a positive finite limit. The entry laws may differ across coordinates and dimensions, and we assume uniform integrability of the fourth powers of the rescaled entries. For entries sampled independently from a fixed law, a finite fourth moment suffices, with no requirement for higher moments. The nonlinearities are continuously differentiable and globally Lipschitz, with bounded continuous first state derivatives. The signal and noise have independent and identically distributed coordinates with finite moments of all orders, and are independent of each other and of the matrix. At every fixed number of iterations, empirical averages of pseudo-Lipschitz functions of order two converge in probability to the limits specified by Gaussian state evolution. These averages include empirical second moments and products of trajectory coordinates. The proof combines an aggregate bound on third derivatives with randomized Lindeberg replacement and operator-norm truncation, without matching third moments or invoking a non-Gaussian universality theorem. We also show that independence of complementary projected components for every deterministic rank-one constraint pair characterizes Gaussian entries. This identifies the limitation of extending Gaussian conditioning identities to independent non-Gaussian matrices. Finally, a variance-profile example shows that average variance normalization alone does not ensure the homogeneous state evolution limits.
\end{abstract}



\section{Introduction}\label{sec:intro}

For the linear observation model
\begin{equation}\label{eq:model}
\by=\bA\bx_0+\bw,
\end{equation}
approximate message passing (AMP) alternates matrix--vector products, coordinatewise nonlinearities, and Onsager corrections. State evolution characterizes the limiting empirical distribution of the iterates via a Gaussian covariance recursion. Bayati and Montanari \cite{BM11} proved this for independent and identically distributed (i.i.d.)\ Gaussian matrices using Gaussian conditioning. Related conditioning methods appear in \cite{Bolthausen14}.

Universality asks whether state evolution holds when other entry laws replace the Gaussian entries. The first matrix multiplication involves sums of independent terms, but later updates reuse the matrix and depend on earlier iterates. A central limit theorem for the first multiplication does not apply to the full recursion.

This paper proves state evolution for rectangular matrices with independent, centered, variance-normalized entries. The entry laws may depend on indices and dimensions. The tail assumption is fourth-power uniform integrability. For i.i.d.\ entries sampled from a fixed law, a finite fourth moment suffices. Under jointly $C^1$, globally Lipschitz nonlinearities with bounded continuous first derivatives, and the side-data assumptions of Section~\ref{sec:algointro}, empirical averages of order-two pseudo-Lipschitz functions, denoted $\PL(2)$, converge in probability to their Gaussian state evolution limits at each iteration. These functions satisfy a Lipschitz bound whose coefficient grows linearly with the argument norms (Definition~\ref{def:PL}). The class includes squared errors and products of trajectory coordinates.

Earlier AMP universality theorems use different matrix and algorithm hypotheses. Bayati, Lelarge, and Montanari \cite{BLM15} analyze polynomial iterations with symmetric sub-Gaussian matrices. Chen and Lam \cite{CL21} analyze iterations with continuously differentiable nonlinearities and Lipschitz first derivatives. Han \cite{Han25} analyzes Lipschitz nonlinearities. Both \cite{CL21, Han25} assume sub-Gaussian matrices. Wang, Zhong, and Fan \cite{WZF} apply tensor networks, allowing variance profiles under dimension-dependent conditions on moments of all orders. Related results evaluate non-separable iterations \cite{LWF25}, polynomial AMP via tree pairings \cite{Kogan25}, first-order methods via a dynamical cavity method \cite{DGPZ26} or diagrammatic expansions \cite{GJKP26}, and non-symmetric or elliptic matrices with variance and correlation profiles \cite{GHN, GHNElliptic}. Gaussian variance-profile models \cite{JM13, BHX}, rotationally invariant models \cite{RSF19, Fan22}, and semirandom models \cite{DLS23} require distinct assumptions.

Our main result, Theorem~\ref{thm:main}, replaces sub-Gaussian or all-moment assumptions with fourth-power uniform integrability for independent entries. It permits entry laws lacking moments above order four. This tail condition and the entry-moment condition in \cite{WZF} do not nest. The sparse ensemble in Example~\ref{ex:sparse} satisfies the latter but violates fourth-power uniform integrability. In contrast, the Pareto ensemble in Example~\ref{ex:pareto} satisfies the former but violates the latter. These examples compare the entry-moment conditions. Section~\ref{sec:scope} compares the nonlinearity, side-data, and matrix-structure assumptions.

Nguyen and Kim \cite{NK23} propose state evolution for this independent-entry model under a uniform moment bound of order greater than four. \cite[Proposition~3]{NK23} gives an affine orthogonal-projection formula, which \cite[Remark~1]{NK23} then uses to characterize a conditional distribution. Such a conditional-law representation needs independence of the complementary projected components. Independent non-Gaussian entries do not in general supply this property. Section~\ref{sec:gaussian} isolates the deterministic projection formula from the component independence that identity requires. Theorem~\ref{thm:sharpcond} proves this independence holds for all deterministic rank-one constraint pairs if and only if the entries are Gaussian. Because this step is invalid for non-Gaussian entries, their argument does not prove the extension. This paper proves the result under fourth-power uniform integrability without non-Gaussian conditioning.

Section~\ref{sec:setting} details the assumptions and main results. Section~\ref{sec:gaussian} analyzes Gaussian conditioning and proves Theorem~\ref{thm:sharpcond}. Sections~\ref{sec:universality} and \ref{sec:assembly} prove Theorem~\ref{thm:main}. Section~\ref{sec:sharp} shows average variance normalization is insufficient. Section~\ref{sec:upgrades} discusses almost-sure convergence and moment-boundary conditions. Section~\ref{sec:scope} compares the hypotheses with other universality theorems.

\section{Setup and Main Results}\label{sec:setting}

\subsection{Notation}

Let $\N$ denote the natural numbers. For a vector $\bv = \sum_{i=1}^d v_i\be_i\in\R^d$ with standard basis $\{\be_i\}_{1\leq i\leq d}$, let $\norm{\bv}$ denote the Euclidean norm, $\norm{\bv}_p$ the $\ell^p$ norm, and $\langle\bv\rangle=\frac1d\sum_i v_i$ the empirical mean. For $\bv,\bu\in\R^d$, $f(\bv) := \sum_{i=1}^d f(v_i)\be_i$ denotes the component-wise application of a scalar function $f$, and $\bu\odot\bv := \sum_{i=1}^d u_i v_i\be_i$ denotes the component-wise product.

For a matrix $\bA\in \mathbb{R}^{n\times N}$, let $\norm{\bA}$ denote the operator norm, $\norm{\bA}_F$ the Frobenius norm, and $\bA^*$ the transpose. Let $n$ denote the number of rows, $N$ the number of columns, and the aspect ratio sequence be $\rho_n=n/N\to\rho\in(0,\infty)$ as $n\to \infty$ (from here on, we omit $n\to \infty$ for brevity).  Let $C,c$ denote generic positive constants that vary by context\footnote{Context dependence includes the fixed horizon, denoisers, aspect ratio, side-information laws, or moment bounds, but excludes $n$ and $N$ unless stated.}. 

For $T\in\N\cup\{0\}$, $t\le T$ denotes $t\in\{0,\ldots,T\}$. Let $Z$ denote a standard Gaussian variable and $\bG$ an $n\times N$ matrix with i.i.d.\ $\mathcal N(0,1/n)$ entries. For a random variable $X$ and $k\ge1$, $L^k$ denotes the space of random variables with $\E\abs{X}^k<\infty$, normed by $\norm{X}_{L^k}:=(\E\abs{X}^k)^{1/k}$. A family is uniformly bounded in $L^k$ if $\sup\norm{X}_{L^k}<\infty$. 

Let $\xrightarrow{\PP}$ denote convergence in probability and $\xrightarrow{\text{a.s.}}$ almost-sure convergence. For random variables $Y_n$ and positive scalars $r_n$, $Y_n=O_\PP(r_n)$ denotes a sequence bounded in probability, meaning $\lim_{K\to\infty}\sup_n\PP(\abs{Y_n}>Kr_n)=0$, and $Y_n=o_\PP(1)$ denotes $Y_n\xrightarrow{\PP}0$. For deterministic sequences, $a_n = O(b_n)$ implies $\limsup_{n\to\infty} \abs{a_n}/b_n < \infty$, and $a_n = o(b_n)$ implies $\lim_{n\to\infty} a_n/b_n = 0$.

Primes $(\prime)$ denote independent copies or alternate entities (e.g., $\bA', \bG', X'$). In some contexts, they denote derivatives. The $k$-th derivative of a real-valued $F$ is $F^{(k)}$. For functions on $\R^d$, $\partial_j^{k}$ denotes the $k$-th derivative in the $j$-th argument.

\begin{definition}[Pseudo-Lipschitz functions]\label{def:PL}
A function $\phi:\R^d\to\R$ is pseudo-Lipschitz of order $2$, denoted $\phi\in\PL(2)$, if there exists $L<\infty$ such that $\abs{\phi(\bx)-\phi(\by)}\le L\left(1+\norm{\bx}+\norm{\by}\right)\norm{\bx-\by}$ for all $\bx,\by\in\R^d$.
\end{definition}

Let $C_b^3(\R^d)$ denote the space of bounded, thrice continuously differentiable functions whose partial derivatives of orders $1$, $2$, and $3$ are bounded. Let $\|\phi\|_\infty:=\sup_{\bx\in\R^d}|\phi(\bx)|$ denote the supremum norm. Let $C_b^{3,1}(\R^d)$ denote the subclass of $C_b^3(\R^d)$ whose third-order partial derivatives are globally Lipschitz. Constants depending on a test function may depend on its supremum norm, its derivative bounds through order $3$, and the Lipschitz constants of its third-order partial derivatives.

\subsection{The algorithm and the assumptions}\label{sec:algointro}

This analysis considers the general two-stage AMP recursion for linear observation models, expressed in the coordinates of the true signal $\bx_0$ and the measurement noise $\bw$.

The algorithm alternates between the signal space $\R^N$ and the measurement space $\R^n$ via the matrix $\bA\in \R^{n\times N}$ and nonlinear maps $f_t, g_t$. The iterates are decoupled via Onsager corrections $\lambda_t \bm^{t-1}$ and $\xi_t \bq^t$. 

The AMP analysis recursion generates state vectors $\bq^t\in\R^N$, $\bb^t,\bm^t\in\R^n$, and $\bh^{t+1}\in\R^N$ for $t\ge0$ via
\begin{equation}\label{eq:amp}
\begin{aligned}
\bq^t&=f_t(\bh^t,\bx_0), &\quad \bb^t&=\bA \bq^t-\lambda_t \bm^{t-1},\\
\bm^t&=g_t(\bb^t,\bw), &\quad \bh^{t+1}&=\bA^* \bm^t-\xi_t \bq^t,
\end{aligned}
\end{equation}
initialized at $\bh^0=\mathbf 0_N$ and $\bm^{-1}=\mathbf 0_n$. The maps $f_t, g_t$ act entrywise. The empirical Onsager coefficients are
\begin{equation}\label{eq:onsager}
\lambda_t=\rho_n^{-1}\,\langle \partial_1 f_t(\bh^t,\bx_0)\rangle, \qquad \xi_t=\langle \partial_1 g_t(\bb^t,\bw)\rangle.
\end{equation}

\begin{assumption}[Measurement ensemble $\cI$]\label{ass:matrix}
Let $\cI$ denote the set of matrices $\bA=(A_{ij})\in\R^{n\times N}$ with $A_{ij}=U_{ij}/\sqrt n$, where the collection $\{U_{ij}\}$ contains independent random variables satisfying
\[
\E U_{ij}=0,\qquad \E U_{ij}^2=1\quad\text{for all }i,j,n,
\]
and
\[
M_{2+2\alpha}:=\sup_{n\ge1}\sup_{i,j}\E\abs{U_{ij}}^{2+2\alpha}<\infty
\qquad\text{for some }\alpha>1 .
\]
The sequence of aspect ratios $\rho_n=n/N$ converges to a limit $\rho\in(0,\infty)$.
\end{assumption}

Define $\mu_{ij} := U_{ij}$ and $\nu := M_{2+2\alpha}$. Because $\rho_n \to \rho$, $\rho_n$ is bounded away from zero and infinity for large $n$. Constants depending on these bounds are denoted $C_\rho$.

\begin{assumption}[Fourth-moment uniform integrability]\label{ass:UI}
$\bA=(\mu_{ij}/\sqrt n)$ has independent entries with $\E\mu_{ij}=0$, $\E\mu_{ij}^2=1$, $\sup_{n,i,j}\E\mu_{ij}^4<\infty$, and
\[
\nu_4(K):=\sup_{n}\sup_{i,j}\E\left[\mu_{ij}^4\,;\,\abs{\mu_{ij}}>K\right]\
\xrightarrow[K\to\infty]{}\ 0 .
\]
\end{assumption}

Assumption~\ref{ass:matrix} implies Assumption~\ref{ass:UI}, since $\nu_4(K)\le\nu K^{-(2\alpha-2)}$. Assumption~\ref{ass:UI} also holds for entries sampled from a fixed centered, unit-variance law with a finite fourth moment.

\begin{assumption}[Nonlinearities]\label{ass:functions}
For $t\ge0$, $f_t,g_t:\R^2\to\R$ are jointly continuously differentiable and globally Lipschitz. The partial derivatives $\partial_1 f_t$ and $\partial_1 g_t$ are bounded. There exists $L_t\ge1$ bounding the Lipschitz constants, the supremum norms $\|\partial_1 f_t\|_\infty$ and $\|\partial_1 g_t\|_\infty$, and the origin evaluations $|f_t(0,0)|$ and $|g_t(0,0)|$.
\end{assumption}

This bound implies
\[
|f_t(0,y)|+|g_t(0,y)| \le 2L_t(1+|y|) \quad \text{for all }y\in\R.
\]
The derivatives $\partial_1 f_t$ and $\partial_1 g_t$ are bounded but not required to be globally Lipschitz.

Certain proofs invoke higher-order regularity. When invoked, we assume $x \mapsto f_t(x,y)$ and $x \mapsto g_t(x,y)$ are thrice continuously differentiable, with uniformly bounded second and third partial derivatives, and globally Lipschitz third partial derivatives uniformly in $y$. There exists $J_t<\infty$ such that
\begin{equation}\tag{S3}\label{eq:smoothcalc}
\begin{aligned}
\sup_{x,y}\sum_{k=2}^3\left(|\partial_1^k f_t(x,y)|+|\partial_1^k g_t(x,y)|\right)&\le J_t,\\
|\partial_1^3 f_t(x,y)-\partial_1^3 f_t(x',y)|
+|\partial_1^3 g_t(x,y)-\partial_1^3 g_t(x',y)|&\le J_t|x-x'|.
\end{aligned}
\end{equation}
Constants in results invoking \eqref{eq:smoothcalc} may depend on $\max_{t\le T}J_t$.

\begin{assumption}[Side data and initialization]\label{ass:signal}
The vectors $\bx_0$ and $\bw$ have i.i.d.\ entries with laws $\mu_{X_0}$ and $\mu_W$. They are independent of each other and $\bA$. The random variables $X_0 \sim \mu_{X_0}$ and $W \sim \mu_W$ have finite moments of all orders. The recursion is initialized at $\bq^0=f_0(\mathbf 0,\bx_0)$.
\end{assumption}

Convergence in probability depends on the marginal law in each dimension, and requires no inter-dimensional coupling. Almost-sure statements assume realization on a common probability space. Lemma~\ref{lem:slln} accommodates dependence across rows of the array.

\subsection{Preliminaries}\label{sec:toolbox}

This subsection states probabilistic and matrix-analytic estimates for subsequent proofs. Proofs are in Appendix~\ref{app:toolbox}. All random variables are real-valued.

\begin{lemma}[Even-moment Rosenthal inequality \cite{Rosenthal70}]\label{lem:rosenthal}
Let $X_1,\dots,X_d$ be independent random variables with $\E X_l=0$ and finite $2p$-th moments for $p\in\N$. Then
\[
\E\left(\sum_{l=1}^dX_l\right)^{2p}
\ \le\ C_p\left[\left(\sum_{l=1}^d\E X_l^2\right)^{p}+\sum_{l=1}^d\E X_l^{2p}\right],
\]
where $C_p:=(8p)^{2p}$. 
\end{lemma}

The constant $C_p = (8p)^{2p}$ is derived in Appendix~\ref{app:toolbox}\footnote{The constant $C_p:=(8p)^{2p}$ is sufficient but not optimal.}. 

\begin{lemma}[Fourth-moment law of large numbers]\label{lem:slln}
Almost-sure convergence holds under bounded fourth central moments.
\begin{enumerate}
    \item[\textup{(a)}] \textbf{Sequences:} Let $(Y_k)_{k\ge1}$ be independent random variables with $\E Y_k=y$ and $\sup_k\E(Y_k-y)^4<\infty$. Then $\frac1d\sum_{k=1}^d Y_k\to y$ almost surely as $d\to\infty$.
    \item[\textup{(b)}] \textbf{Triangular arrays:} Let $(Y_{d,k})_{d\ge1, 1\le k\le d}$ be a triangular array where, for each fixed row $d$, the variables $(Y_{d,k})_{k=1}^d$ are mutually independent. If $\E Y_{d,k}=y$ and $\sup_{d,k}\E(Y_{d,k}-y)^4<\infty$, then $\frac1d\sum_{k=1}^d Y_{d,k}\to y$ almost surely as $d\to\infty$.
\end{enumerate}
\end{lemma}

\begin{theorem}[Lata{\l}a's norm bound {\cite[Theorem~2]{Latala05}}]\label{thm:latala}
There exists a constant $C_\ast$ such that for any random matrix $\mathbf D = (D_{ij}) \in \R^{n \times N}$ with mutually independent, mean-zero entries,
\[
\E\norm{\mathbf D} \le C_\ast \Biggl( \max_{1\le i\le n}\left(\sum_{j=1}^N \E D_{ij}^2\right)^{1/2} + \max_{1\le j\le N}\left(\sum_{i=1}^n \E D_{ij}^2\right)^{1/2} + \left(\sum_{i=1}^n \sum_{j=1}^N \E D_{ij}^4\right)^{1/4} \Biggr).
\]
\end{theorem}

\subsection{State evolution}

The AMP recursion alternates multiplication by \(A\) and \(A^*\) with coordinatewise nonlinearities \(f_t\) and \(g_t\). Repeated multiplication by $\bA$ induces correlations between the iterates and the matrix, invalidating standard central limit arguments. The Onsager corrections compensate for the leading effects of matrix reuse. For Gaussian matrices, orthogonal conditioning formally decouples the iterates. The deterministic covariance recursion of this decoupled system is the state evolution. The universality argument transfers this Gaussian limit to the non-Gaussian ensembles. We define the limiting Gaussian variables and their covariance sequence as follows.

\begin{definition}[Two-time state evolution]\label{def:SE}
Let $(X_0,W) \sim \mu_{X_0} \otimes \mu_W$. Define sequences $(B^0,\dots,B^T)$ and $(H^1,\dots,H^{T+1})$ as centered, jointly Gaussian families independent of $(X_0,W)$, with covariances initialized at $H^0 := 0$:
\begin{equation}\label{eq:SE}
\begin{aligned}
\E\bigl[B^rB^s\bigr] &= \frac{1}{\rho}\E\bigl[f_r(H^r,X_0)f_s(H^s,X_0)\bigr],\\
\E\bigl[H^{r+1}H^{s+1}\bigr] &= \E\bigl[g_r(B^r,W)g_s(B^s,W)\bigr],
\end{aligned}
\end{equation}
for $0\le r,s\le T$. The state variances are
\[
\sigma_t^2 := \E\bigl[(B^t)^2\bigr] = \frac{1}{\rho}\E\bigl[f_t(H^t,X_0)^2\bigr], \qquad
\tau_t^2 := \E\bigl[(H^{t+1})^2\bigr] = \E\bigl[g_t(B^t,W)^2\bigr].
\]
The deterministic Onsager coefficients are
\begin{equation}\label{eq:onsagerdet}
\bar{\lambda}_t := \frac{1}{\rho}\E\bigl[\partial_1 f_t(H^t,X_0)\bigr],
\qquad
\bar{\xi}_t := \E\bigl[\partial_1 g_t(B^t,W)\bigr].
\end{equation}
Define the limit trajectory vectors
\begin{equation}\label{eq:Vstar}
V^\ast_H := (H^1,\dots,H^{T+1},X_0) \in \R^{T+2}, \qquad V^\ast_B := (B^0,\dots, B^T,W) \in \R^{T+2}.
\end{equation}
\end{definition}

\begin{remark}\label{rem:SE}
Equation~\eqref{eq:SE} defines the state evolution \cite{BM11}. The covariance of $B$ at step $t$ depends on the joint law of $H$ up to step $t$. This inductive construction defines a consistent positive-semidefinite multivariate Gaussian distribution.
\end{remark}

The families $(B^t)$ and $(H^t)$ are mutually independent. Positivity conditions (e.g., $\sigma_0^2 > 0$) are assumed only when required.

\subsubsection*{Fixed-level truncation}
For a fixed cutoff $M > 0$, define the truncated and standardized entries
\[
U_{ij}^{(M)} := \frac{U_{ij}\ind\{|U_{ij}|\le M\}-\E\bigl[U_{ij}\ind\{|U_{ij}|\le M\}\bigr]}{\sqrt{\Var\bigl(U_{ij}\ind\{|U_{ij}|\le M\}\bigr)}},
\]
and define the matrix $\bAh_M:= n^{-1/2}(U_{ij}^{(M)})$. For large $M$, the variance denominator is non-zero. The limit $n \to \infty$ is taken at fixed $M$. Lemma~\ref{lem:truncremoval} details the approximation estimates for this ensemble.

\subsection{Main results}

\begin{theorem}[State evolution under fourth-power uniform integrability]\label{thm:main}
Fix $T\in\N\cup\{0\}$. Assume the nonlinearities and side-data satisfy Assumptions~\ref{ass:functions} and \ref{ass:signal}, and the measurement matrix satisfies Assumption~\ref{ass:matrix} or the uniform integrability condition of Assumption~\ref{ass:UI}. For every pseudo-Lipschitz test function $\phi\in\PL(2)$ on $\R^{T+2}$,
\begin{align}
\frac{1}{N}\sum_{k=1}^{N}\phi\bigl(h^1_k,\dots,h^{T+1}_k,x_{0k}\bigr)
&\;\xrightarrow{\ \PP\ }\;
\E\bigl[\phi(V^\ast_H)\bigr], \label{eq:mainh}\\
\frac{1}{n}\sum_{i=1}^{n}\phi\bigl(b^0_i,\dots,b^{T}_i,w_i\bigr)
&\;\xrightarrow{\ \PP\ }\;
\E\bigl[\phi(V^\ast_B)\bigr], \label{eq:mainb}
\end{align}
as $n\to\infty$. 

This convergence holds in $L^1$ for
\begin{enumerate}
    \item[\textup{(a)}] The untruncated ensemble $\bA$, for bounded $\phi$.
    \item[\textup{(b)}] The iterates computed with $\bAh_M$ in place of $\bA$, for every sufficiently large fixed $M$ and every $\phi\in\PL(2)$.
\end{enumerate}
If the sequence of ensembles is realized on a common probability space, there exists a deterministic subsequence along which \eqref{eq:mainh} and \eqref{eq:mainb} hold almost surely for fixed $\phi$.
\end{theorem}

Theorem~\ref{thm:main} implies that the empirical mean-square statistics converge in probability to the state evolution limits $\langle(\bh^{t+1})^2\rangle\to\tau_t^2$ and $\langle(\bb^{t})^2\rangle\to\sigma_t^2$ for $t\le T$. 

For unbounded $\PL(2)$ tests, Lemma~\ref{lem:empmom} and convergence in probability give $L^1$ convergence at fixed truncation levels. Theorem~\ref{thm:as} governs almost-sure convergence. By Theorem~\ref{thm:as}(b), a deterministic almost-sure subsequence exists for any countable collection of in-probability limits. 

Theorem~\ref{thm:swap} compares expectations of smooth empirical observables under Gaussian matrix replacement. Lemma~\ref{lem:varbound} extends the comparison to $C^1$ nonlinearities and transfers concentration through the first two raw moments. Lemma~\ref{lem:truncremoval} bounds $\E\norm{\bA-\bAh_M}$ uniformly in $n$ and transfers trajectory limits in probability. Proposition~\ref{prop:fourth} compares the centered fourth moments of the non-Gaussian and Gaussian statistics.

Subsequent sections analyze the convergence hypotheses. Theorem~\ref{thm:sharpvar} shows that aggregate variance normalization does not imply state evolution. Section~\ref{sec:boundary} contrasts the fourth-power condition with the weaker condition sufficient for the first row update. Theorem~\ref{thm:as} specifies conditions for almost-sure convergence.

\begin{theorem}[Universal deterministic conditioning characterizes Gaussian entries]\label{thm:sharpcond}
Let $n, N\ge2$. Let $\bA \in \R^{n \times N}$ have independent, mean-zero entries with variance $1/n$. For one-dimensional subspaces in $\R^n$ and $\R^N$, let $\Pi$ denote the orthogonal projection onto their complements (Section~\ref{sec:gaussian}). The projections $\Pi(\bA)$ and $\Pi^\perp(\bA)$ are independent for all such pairs if and only if the entries of $\bA$ are i.i.d.\ $\mathcal N(0,1/n)$.
\end{theorem}

Theorem~\ref{thm:sharpcond} characterizes Gaussian entries by independence of the two orthogonal projections for every rank-one constraint pair. Under its independent-entry, centering, and variance assumptions, every non-Gaussian ensemble fails this property for some pair. Proposition~\ref{prop:condfails}(ii) shows that distinct non-Gaussian ensembles can share all mixed entry moments, and both fail the projection-independence property. Section~\ref{sec:gaussian} proves these statements.

\section{Gaussian Matrices and Orthogonal Conditioning}
\label{sec:gaussian}

\subsection{Gaussian conditioning via orthogonal projections}

Fix matrices $\bM\in\R^{n\times t}$ and $\bQ\in\R^{N\times t}$ with full column rank. Let $\PM$ and $\PQ$ denote the orthogonal projectors onto their column spaces
\[
\PM = \bM(\bM^*\bM)^{-1}\bM^*, \qquad \PQ = \bQ(\bQ^*\bQ)^{-1}\bQ^*.
\]
Define $\PMp=\mathbf I-\PM$ and $\PQp=\mathbf I-\PQ$. For $\mathbf C \in \R^{n \times N}$, define the orthogonal projections
\[
\Pi(\mathbf C):=\PMp\,\mathbf C\,\PQp ,\quad \text{ and }\quad\Pi^\perp(\mathbf C):=\mathbf C-\Pi(\mathbf C)=\PM\mathbf C+\PMp\mathbf C\PQ .
\]
By definition, $\Pi(\mathbf C)\bQ = \mathbf 0$ and $\bM^*\Pi(\mathbf C) = \mathbf 0$. The projection $\Pi^\perp(\mathbf C)$ is determined by $\mathbf C\bQ$ and $\bM^*\mathbf C$.

\begin{lemma}[Orthogonal projection onto $\cV$]\label{lem:projector}
For the Frobenius inner product $\inner{\mathbf S}{\mathbf T}=\tr(\mathbf S^*\mathbf T)$, $\Pi$ is the orthogonal projector onto $\cV=\{\mathbf H\in\R^{n\times N}:\PM\mathbf H=\mathbf 0,\ \mathbf H\PQ=\mathbf 0\}$.
\end{lemma}

\begin{proof}
Because $\PMp$ and $\PQp$ are symmetric and idempotent, $\Pi(\mathbf C) = \PMp\mathbf C\PQp$ is self-adjoint and idempotent. It is the orthogonal projector onto its range, which is $\cV$.
\end{proof}

\begin{proposition}[Gaussian conditioning]\label{prop:bolthausen}
Let $\bG \in \R^{n \times N}$ have i.i.d.\ $\mathcal N(0,1/n)$ entries. The projections $\Pi(\bG)$ and $\Pi^\perp(\bG)$ are independent. Conditional on the $\sigma$-algebra generated by $\bG^*\bM$ and $\bG\bQ$, the distribution is
\begin{equation}\label{eq:bolthausen}
\bG\ \big|\ \sigma(\bG^*\bM,\bG\bQ)\ \ \eqd\ \ \Pi^\perp(\bG)\ +\ \Pi(\bG')\,,
\end{equation}
where $\bG'$ is an independent copy of $\bG$.

For observations $\bG^*\bM=\bX$ and $\bG\bQ=\bY$ satisfying $\bM^*\bY=\bX^*\bQ$, the conditional distribution is $\bB+\Pi(\bG')$, where
\[
\begin{aligned}
\bB &= \bB_1 + \bB_2 - \bB_3 \,, \quad \text{where} \\
\bB_1 &= \bY(\bQ^*\bQ)^{-1}\bQ^* \,, \\
\bB_2 &= \bM(\bM^*\bM)^{-1}\bX^* \,, \\
\bB_3 &= \bM(\bM^*\bM)^{-1}\bX^*\bQ(\bQ^*\bQ)^{-1}\bQ^* .
\end{aligned}
\]
\end{proposition}

\begin{proof}
Because $\bG$ has i.i.d.\ Gaussian entries, its distribution is isotropic. By Lemma~\ref{lem:projector}, $\Pi$ and $\Pi^\perp$ are orthogonal projectors. The projections $\Pi(\bG)$ and $\Pi^\perp(\bG)$ are uncorrelated jointly Gaussian matrices, hence independent.

We have $\sigma(\bG^*\bM,\bG\bQ)=\sigma(\Pi^\perp(\bG))$. The observations define the projection via $\PM\bG = \bM(\bM^*\bM)^{-1}(\bG^*\bM)^*$ and $\bG\PQ = (\bG\bQ)(\bQ^*\bQ)^{-1}\bQ^*$, yielding $\Pi^\perp(\bG) = \PM\bG + \PMp(\bG\PQ)$. Conversely, $\Pi^\perp(\bG)$ defines the observations via $\bG^*\bM=(\Pi^\perp\bG)^*\bM$ and $\bG\bQ=(\Pi^\perp\bG)\bQ$, proving the $\sigma$-algebra equivalence. Because $\Pi(\bG)$ is independent of $\Pi^\perp(\bG)$, it is independent of this $\sigma$-algebra, proving \eqref{eq:bolthausen}.

For constraints $\bG^*\bM=\bX$ and $\bG\bQ=\bY$ satisfying $\bM^*\bY=\bX^*\bQ$, substitution gives $\Pi^\perp(\bG) = \bB_2 + (\mathbf I - \PM)\bB_1 = \bB_2 + \bB_1 - \bB_3 = \bB$. Substituting this constant into \eqref{eq:bolthausen} completes the proof.
\end{proof}

\subsection{Geometric characterization of the Gaussian ensemble}

\begin{proposition}[Universal projection independence characterizes Gaussian entries]\label{prop:condfails}
\textup{(i)} Let $n,N\ge2$ and let $\bA \in \R^{n \times N}$ have independent centered entries of variance $1/n$. If for every pair of one-dimensional orthogonal projectors parameterized by nonzero vectors $\bM\in\R^n$ and $\bQ\in\R^N$ the components $\Pi(\bA)$ and $\Pi^\perp(\bA)$ are independent, the entries of $\bA$ are i.i.d.\ Gaussian.

\textup{(ii)} There exist non-Gaussian ensembles $\bA,\bAt\in\cI$ with identical mixed entry moments at every order. Moment sequences do not determine the finite-dimensional matrix law, and neither ensemble satisfies the projection-independence property of part \textup{(i)}.
\end{proposition}

\begin{proof}
(i) Independence of $\Pi(\bA)$ and $\Pi^\perp(\bA)$ for all one-dimensional subspaces implies that for every unit vector $\bu\in\R^n$ and $\bv\perp\bu$, the projections $\bu^*\bA\be_j$ and $\bv^*\bA\be_j$ are independent. Setting $\bM=\bu$ and $\bQ=\be_{j'}$ for $j'\neq j$ gives $\PQp\be_j=\be_j$ and $\bu^*\PMp=\mathbf 0$. Thus, $\bv^*\Pi(\bA)\be_j = \bv^*\bA\be_j$ and $\bu^*\Pi^\perp(\bA)\be_j = \bu^*\bA\be_j$.

For column $\ba=\bA\be_j$, setting $\bu=(\be_1+\be_2)/\sqrt2$ and $\bv=(\be_1-\be_2)/\sqrt2$ restricts this independence to entries $X:=a_1$ and $Y:=a_2$. By hypothesis, $X$ and $Y$ are independent, and orthogonal projection implies $X+Y$ and $X-Y$ are independent. By the Kac--Bernstein theorem \cite{Kac39}, $X$ and $Y$ are Gaussian. Iterating across all coordinates proves every entry of $\bA$ is Gaussian.

(ii) The lognormal law is not determined by its moments. Other laws have the same moments \cite{Heyde63}. Centering two such laws at their common mean and dividing by their common standard deviation gives distinct standardized laws $\mu^{(0)}\neq\mu^{(1)}$ with identical moment sequences. For random matrices with these independent entries, every mixed moment factorizes into its marginals
\[
\E\prod_{(i,j)\in S}A_{ij}^{k_{ij}}=\prod_{(i,j)\in S}\E A_{ij}^{k_{ij}}.
\]
Because marginal moments coincide, all joint moments match. The generating laws imply $\bA\not\eqd\bAt$. Neither law is Gaussian, since both have a nonzero third moment of a standardized lognormal law, so both violate the independence property of part (i). 
\end{proof}

Theorem~\ref{thm:sharpcond} requires projection independence for every rank-one constraint pair. Restricting the condition to coordinate subspaces does not characterize Gaussian entries, since the corresponding projections depend on disjoint sets of independent entries.

\begin{proof}[Proof of Theorem~\ref{thm:sharpcond}]
Proposition~\ref{prop:bolthausen} and Proposition~\ref{prop:condfails}(i) prove the forward and converse implications. Independence of $\Pi(\bA)$ and $\Pi^\perp(\bA)$ means either conditional component retains its unconditional law after conditioning on the other. Conditioning on $\Pi^\perp(\bA)$ gives the independent-copy representation in Proposition~\ref{prop:bolthausen}. Universal quantification over subspaces is necessary because for coordinate subspaces, orthogonal projectors map to disjoint index sets of the matrix. These projections are independent for any ensemble with independent entries.
\end{proof}

\begin{theorem}[Gaussian state evolution]\label{thm:BM}
Under Assumptions~\ref{ass:functions}--\ref{ass:signal}, the Gaussian iterates satisfy \eqref{eq:mainh}--\eqref{eq:mainb} in probability. The auxiliary iterates using the deterministic coefficients \eqref{eq:onsagerdet} satisfy the same conclusions in probability. For bounded pseudo-Lipschitz test functions, the empirical statistics converge in every finite $L^p$.
\end{theorem}

\begin{proof}
Empirical convergence of the Gaussian iterates follows from \cite[Theorem~2, Lemma~1(b,c)]{BM11}, where Lemma~1(b) gives joint empirical convergence. Assumption~\ref{ass:signal} and Lemma~\ref{lem:slln} verify the empirical moment conditions. Appendix~\ref{app:traj} treats singular covariances.

Because $\partial_1 f_t$ and $\partial_1 g_t$ are bounded continuous functions, Gaussian empirical convergence extends to these test functions via approximation on compact sets, given tight empirical second moments (as in Lemma~\ref{lem:continuoustest}). The empirical coefficients converge in probability to the deterministic limits \eqref{eq:onsagerdet}.

Let $(\bh^t, \bb^t)$ denote true iterates and $(\tilde{\bh}^t, \tilde{\bb}^t)$ denote auxiliary iterates. Fixing the empirical coefficients and subtracting the recursions, Lemma~\ref{lem:stability}, Lemma~\ref{lem:iteratenorms}, and tight Gaussian operator norms imply $\frac{1}{\sqrt{N}}\|\bh^t - \tilde{\bh}^t\|_2 = o_\PP(1)$ and $\frac{1}{\sqrt{n}}\|\bb^t - \tilde{\bb}^t\|_2 = o_\PP(1)$. For pseudo-Lipschitz test functions, Cauchy--Schwarz bounds the difference in empirical statistics by these norms. Because the test functions are bounded, the empirical statistics are bounded. Convergence in probability implies convergence in $L^p$ for finite $p$ by uniform integrability.
\end{proof}

\section{Universality of the Trajectory}\label{sec:universality}

To decouple matrix randomness from the nonlinear empirical averages defining the true Onsager coefficients, this section analyzes an auxiliary recursion parameterized by deterministic sequences. Let $(\bar\lambda_t, \bar\xi_t)$ be arbitrary deterministic scalar sequences bounded uniformly in $n$ and $N$. The auxiliary recursion initialized at $\bh^0=\mathbf{0}$ and $\bm^{-1}=\mathbf{0}$ is defined as
\begin{equation}\label{eq:auxamp}
\begin{aligned}
\bq^t&=f_t(\bh^t,\bx_0),& \bb^t&=\bC\bq^t-\bar\lambda_t \bm^{t-1},\\
\bm^t&=g_t(\bb^t,\bw),&\bh^{t+1}&=\bC^*\bm^t-\bar\xi_t\bq^t.
\end{aligned}
\end{equation}
To streamline notation, we omit the tildes (e.g., $\tilde{\bh}^t$) that previously distinguished these auxiliary iterates from the original iterates. We will reintroduce them from the next section onwards.

For each coordinate $i$, define the trajectory vectors as
\[
\bv_{H,i}=(h_i^1,\ldots,h_i^{T+1},x_{0,i}), \qquad \bv_{B,i}=(b_i^0,\ldots,b_i^T,w_i).
\]
For a given test function $\phi$, the empirical observable evaluated at $\bC$ is
\begin{equation}\label{eq:observable}
H(\bC)=\frac1d\sum_{i=1}^d\phi(\bv_i(\bC)),
\end{equation}
where $d=N$ for the signal-space trajectories $\bv_{H,i}$ and $d=n$ for the measurement-space trajectories $\bv_{B,i}$. We denote $D=\max(n,N)$. The partial derivatives $\partial_{ab}$ are taken with respect to the matrix coordinate $C_{ab}$ itself. 

\subsection{Energy and perturbation estimates}

For any vector $\bu$, let
\[
\|\bu\|_{\rm av}:=(\dim\bu)^{-1/2}\|\bu\|_2,
\]
so that vectors in the signal and measurement spaces are normalized by $N^{1/2}$ and $n^{1/2}$ respectively. Define the parameters
\[
S := 1+\|\bx_0\|_{\rm av}+\|\bw\|_{\rm av},\qquad
K_C := 1+\|\bC\|+\max_{t\le T}(|\bar\lambda_t|+|\bar\xi_t|).
\]
Constants in this subsection may depend on the time horizon $T$, the common Lipschitz and origin bounds of the maps, a common bound on the deterministic Onsager coefficients, and fixed upper and lower bounds on the aspect ratios. They do not depend on $n$ or $N$.

\begin{lemma}[Energy bound]\label{lem:iteratenorms}
Let the coefficients $(\bar\lambda_t, \bar\xi_t)$ be uniformly bounded. There exist constants $C_T, a_T > 0$ such that the iterates of \eqref{eq:auxamp} satisfy
\[
\max_{t\le T}\max_{\bu \in \{\bh^{t+1}, \bb^t, \bm^t, \bq^t\}}\|\bu(\bC)\|_{\rm av} \le C_TK_C^{a_T}S.
\]
Bounded first state derivatives imply bounded empirical coefficients. This bound extends to the true recursion by replacing the coefficient term in $K_C$ with its deterministic bound.
\end{lemma}

\begin{proof}
By the Lipschitz continuity of $f_t$ and $g_t$, the normalized norms satisfy
\[
\|\bq^t\|_{\rm av} \le L(\|\bh^t\|_{\rm av}+\|\bx_0\|_{\rm av}+1),\qquad
\|\bm^t\|_{\rm av} \le L(\|\bb^t\|_{\rm av}+\|\bw\|_{\rm av}+1).
\]
The linear updates satisfy $\|\bC\bq^t\|_{\rm av} \le C_\rho\|\bC\|\|\bq^t\|_{\rm av}$ and $\|\bC^*\bm^{t-1}\|_{\rm av} \le C_\rho\|\bC\|\|\bm^{t-1}\|_{\rm av}$. Iterating these bounds from $\bh^0=\mathbf{0}$ and $\bm^{-1}=\mathbf{0}$ gives the result.
\end{proof}

\begin{lemma}[Deterministic stability]\label{lem:stability}
Let $(\bC,f_t,g_t,\bar\lambda_t,\bar\xi_t)$ and $(\bC',f'_t,g'_t,\bar\lambda'_t,\bar\xi'_t)$ generate auxiliary recursions from identical side data $(\bx_0,\bw)$. For $0\le t\le T$, assume that all four maps are jointly $L$-Lipschitz and that their absolute values at $(0,0)$ are at most $L$. Assume also that
\[
\max_{t\le T}\max\{ |\bar\lambda_t|,|\bar\xi_t|, |\bar\lambda'_t|,|\bar\xi'_t|\}\le B
\]
and
\[
\max_{t\le T}\max\{\|f_t-f'_t\|_\infty,\|g_t-g'_t\|_\infty\}\le\delta,\qquad
\max_{t\le T}\max\{|\bar\lambda_t-\bar\lambda'_t|,|\bar\xi_t-\bar\xi'_t|\}\le e.
\]
There exists a polynomial $P_T$, whose coefficients may depend on $T, L, B$ and fixed upper and lower aspect-ratio bounds, but not on $n$ or $N$, such that
\begin{equation}\label{eq:stabilitynew}
\max_{t\le T}\max_{\bu \in \{\bh^{t+1},\bb^t,\bm^t,\bq^t\}}\|\bu(\bC)-\bu'(\bC')\|_{\rm av}\le P_T(1+\|\bC\|+\|\bC'\|)\bigl\{S(\|\bC-\bC'\|+e)+\delta\bigr\}.
\end{equation}
\end{lemma}

\begin{proof}
Define the discrepancy $\Delta\bu = \bu - \bu'$. The triangle inequality and $L$-Lipschitz continuity imply $\|\Delta\bq^t\|_{\rm av} \le L\|\Delta\bh^t\|_{\rm av} + \delta$ and $\|\Delta\bm^t\|_{\rm av} \le L\|\Delta\bb^t\|_{\rm av} + \delta$. Subtracting the linear updates gives
\[
\Delta\bb^t = \bC\Delta\bq^t + (\bC-\bC')\bq'^t - \bar\lambda_t\Delta\bm^{t-1} - (\bar\lambda_t-\bar\lambda'_t)\bm'^{t-1}.
\]
Lemma~\ref{lem:iteratenorms} bounds the energies: $\|\bq'^t\|_{\rm av}, \|\bm'^{t-1}\|_{\rm av} \le C_T K_{C'}^{a_T} S$. Because the normalized norms use the respective vector dimensions,
\[
\|\bC\Delta\bq^t\|_{\rm av}
\le \rho_n^{-1/2}\|\bC\|\|\Delta\bq^t\|_{\rm av}.
\]
The triangle inequality implies
\[
\|\Delta\bb^t\|_{\rm av}
\le \rho_n^{-1/2}\|\bC\|\|\Delta\bq^t\|_{\rm av}
+|\bar\lambda_t|\|\Delta\bm^{t-1}\|_{\rm av}
+C_T K_{C'}^{a_T}S
\bigl(\rho_n^{-1/2}\|\bC-\bC'\|+e\bigr).
\]
The transpose update satisfies the corresponding bound with factor $\rho_n^{1/2}$. Because $\rho_n\to\rho\in(0,\infty)$, both factors are bounded. Absorbing them and iterating over $T$ steps proves \eqref{eq:stabilitynew}.
\end{proof}

\subsection{Coordinate moments by removal of a row or a column}

Let $\bC$ have independent, centered entries satisfying the sub-Gaussian bound
\begin{equation}\label{eq:workingSG}
\E e^{zC_{ab}}\le \exp(V^2z^2/(2n)),\qquad z\in\R,
\end{equation}
for fixed $V$. This allows vanishing variances and sparse configurations. Appendix~\ref{app:toolbox} bounds all finite moments of the operator norm $\|\bC\|$ via an $\varepsilon$-net argument. Assumption~\ref{ass:signal} gives all finite moments for the side data $(x_0, w)$.

\begin{lemma}[Cavity moment bounds]\label{lem:momprop}
For fixed $p\ge1$, the sequence of iterates satisfies
\begin{equation}\label{eq:coordmom}
\sup_n \max_{t\le T} \max_{u \in \{h^{t+1}, b^t, m^t, q^t\}} \max_i \E|u_i(\bC)|^p \le C_{V,T,p}.
\end{equation}
Let $R := 1+\max_{t\le T}\max_{u}\max_i|u_i(\bC)|$. For fixed $c,j\ge 0$ and $\varepsilon>0$:
\begin{equation}\label{eq:Rmoment}
\E[K_C^c R^j] \le C_{V,T,c,j,\varepsilon} n^\varepsilon.
\end{equation}
\end{lemma}

\begin{proof}
By Jensen's and H\"older's inequalities, products of arbitrary powers of $S$ and $K_C$ have bounded expectations. Lemma~\ref{lem:iteratenorms} bounds all finite moments of the normalized global trajectory norms.

Fix a row index $a \in [n]$. Define $\bC^{(a)}$ by annihilating the $a$-th row of $\bC$, giving the decomposition $\bC - \bC^{(a)} = e_a r_a^*$. Let $u^{(a)}$ denote the auxiliary iterates generated by $\bC^{(a)}$. Defining the discrepancy $\Delta u = u - u^{(a)}$, subtracting the recursions gives
\begin{align*}
\Delta b^t &= \bC\Delta q^t + e_a r_a^* q^{t,(a)} - \bar\lambda_t\Delta m^{t-1},\\
\Delta h^{t+1} &= \bC^*\Delta m^t + r_a m_a^{t,(a)} - \bar\xi_t\Delta q^t.
\end{align*}
Iterating the Lipschitz bounds on $\|\cdot\|_2$ gives
\begin{equation}\label{eq:rowforcingbound}
\max_{t\le T}\max_{\bu}\|\bu^t-\bu^{t,(a)}\|_2
\le
P_T(K_C+K_{C^{(a)}})
\sum_{s\le T}\left(|r_a^*q^{s,(a)}|+\|r_a\|_2|m_a^{s,(a)}|\right),
\end{equation}
where $P_T$ depends only on $T$ and $L$. The removed row appears only in
\begin{equation}\label{eq:rowforcing}
\sum_{s\le T}\left(|r_a^*q^{s,(a)}|+\|r_a\|_2|m_a^{s,(a)}|\right).
\end{equation}
By hypothesis, $r_a$ is independent of the $\sigma$-algebra generated by $(\bC^{(a)}, x_0, w)$. Conditioning on this sub-$\sigma$-algebra, \eqref{eq:workingSG} implies
\[
\E\left[|r_a^*q^{s,(a)}|^p \mid \bC^{(a)},x_0,w\right] \le C_p V^p n^{-p/2}\|q^{s,(a)}\|_2^p.
\]
Because $N/n$ is bounded, Lemma~\ref{lem:iteratenorms} bounds the expectation of the right side. Annihilating the $a$-th row, we get $b_a^{s,(a)} = -\bar\lambda_s m_a^{s-1,(a)}$, and the scalar recursion gives $|b_a^{s,(a)}| + |m_a^{s,(a)}| \le C_T(1+|w_a|)$. Because $\|r_a\|_2$ has finite moments, all terms in \eqref{eq:rowforcing} have finite moments. $\bC$ and $\bC^{(a)}$ satisfy sub-Gaussian operator-norm moment bounds, and $\|r_a\|_2$ has moments of all orders. H\"older's inequality applied to \eqref{eq:rowforcingbound} bounds the discrepancy moments. The triangle inequality bounds the coordinate estimates for $b_a^t$ and $m_a^t$.

For column $b \in [N]$, the decomposition $\bC-\bC^{[b]} = c_b e_b^*$ gives the discrepancy terms $c_b q_b^{s,[b]}$ and $e_b c_b^*m^{s,[b]}$. From the zero-column scalar recursion, $h_b^{s+1,[b]}=-\bar\xi_s q_b^{s,[b]}$ and $q_b^{s,[b]}=f_s(h_b^{s,[b]},x_{0b})$, giving the bound $C_T(1+|x_{0b}|)$. Conditional sub-Gaussian moments of $c_b^*m^{s,[b]}$, combined with the global energy bound, the operator-norm moments of $\bC^{[b]}$, and H\"older's inequality, bound $h_b^{t+1}$ and $q_b^t$. The row and column arguments prove \eqref{eq:coordmom}.

A union over the $\mathcal{O}(D)$ trajectory elements bounds $\E[R^q] \le C_{V,T,q}D$ for $q \ge 1$. For $j > 0$ and $\varepsilon > 0$, set $q = \max(1,2j, j/\varepsilon)$. H\"older's and Lyapunov's inequalities give
\[
\E[K_C^c R^j] \le (\E K_C^{2c})^{1/2} (\E R^{2j})^{1/2} \le C(\E R^q)^{j/q} \le Cn^\varepsilon.
\]
The moments of $K_C$ satisfy the $j=0$ case.
\end{proof}

\subsection{Entry derivatives and their sources}

Assume condition~\eqref{eq:smoothcalc} holds. Fix a matrix $\bC$ and an index pair $e=(a,b)$. Let $\partial$ denote $\partial_{ab}$. Define $K := K_C$ and the trajectory bound $R := \max(1, \max_{u,i}|u_i(\bC)|)$.

\begin{lemma}[Differentiated recursion]\label{lem:derivatives}
For $r \in \{1,2,3\}$, the differentiated linear updates are
\begin{equation}\label{eq:Drecursion}
\begin{aligned}
\partial^r b^t &= \bC\partial^r q^t + r e_a(\partial^{r-1} q^t)_b - \bar\lambda_t\partial^r m^{t-1},\\
\partial^r h^{t+1} &= \bC^*\partial^r m^t + r e_b(\partial^{r-1} m^t)_a - \bar\xi_t\partial^r q^t.
\end{aligned}
\end{equation}
For $q = f(h,x_0)$, the differentiated nonlinear updates are
\begin{equation}\label{eq:Dstructure}
\begin{aligned}
\partial q &= \partial_1 f\odot\partial h,\\
\partial^2 q &= \partial_1 f\odot\partial^2 h + \partial_1^2 f\odot(\partial h)^{\odot 2},\\
\partial^3 q &= \partial_1 f\odot\partial^3 h + 3\partial_1^2 f\odot\partial h\odot\partial^2 h + \partial_1^3 f\odot(\partial h)^{\odot 3}.
\end{aligned}
\end{equation}
The analogous expansions hold for $m=g(b,w)$. For every iterate $u$,
\begin{equation}\label{eq:pointvector}
\|\partial^r u\|_2 \le C_T K^{a_T} R^r,\qquad r \in \{1,2,3\}.
\end{equation}
\end{lemma}

\begin{proof}
Leibniz's rule and the identity $\partial\bC=e_a e_b^*$ imply \eqref{eq:Drecursion}, and the chain rule implies \eqref{eq:Dstructure}. Initial derivatives vanish. By induction, the derivatives satisfy a linearized recursion governed by the operators $\bC, \bC^*$ and bounded diagonal coefficient matrices, whose lengths and quantities are bounded by constants depending on $T$. Because $\|x\odot y\|_2 \le \|x\|_\infty\|y\|_2$, iterating these bounds gives \eqref{eq:pointvector}.
\end{proof}

\begin{lemma}[Finite source expansion]\label{lem:finitesources}
At derivative order $r\le3$, every iterate $u$ admits the representation
\begin{equation}\label{eq:sourcepaths}
\partial^r u = \sum_{\nu}L_{u\nu}s_{\nu,e}^{(r)},
\end{equation}
where $s_{\nu,e}^{(r)}$ denotes the inhomogeneous source vectors generated by the nonlinear chain rule \eqref{eq:Dstructure} and the Leibniz matrix insertions \eqref{eq:Drecursion}. The number of terms depends only on $T$. The linear operators $L_{u\nu}$ are independent of the differentiated index $e=(a,b)$, with operator norms bounded by $C_TK^{a_T}$.

At $r=1$, every derivative is a finite sum of terms of the form
\begin{equation}\label{eq:firstsources}
Pe_a\alpha_b \quad\text{and}\quad Qe_b\beta_a.
\end{equation}
The vectors $\alpha,\beta$ and matrices $P, Q$ are independent of $(a,b)$; the vector entries have absolute value at most $R$, and the operator norms of $P, Q$ are bounded by $C_TK^{a_T}$.

At $r=2$, every coordinate derivative is a finite sum of the forms
\begin{equation}\label{eq:secondforms}
U_{ia}\alpha_b\tilde\alpha_b,\quad V_{ib}\beta_a\tilde\beta_a,\quad\alpha_b\beta_a\sum_jL_{ij}P_{ja}Q_{jb},\quad L_{ia}P_{ba}\alpha_b,\quad \text{and}\quad L_{ia}Q_{bb}\beta_a,
\end{equation}
and the forms obtained by interchanging row and column roles. Here $U=L(P\odot\widetilde P)$ and $V=L(Q\odot\widetilde Q)$. All matrices are independent of $(a,b)$, the operator norms of $L,P,Q,\widetilde P,\widetilde Q$ are bounded by $C_TK^{a_T}$, and the entries of $\alpha,\tilde\alpha,\beta,\tilde\beta$ have absolute value at most $R$. The matrices and vectors may differ between summands. Bounded diagonal chain-rule coefficients and numerical factors are absorbed into the propagation matrices.

At $r=3$, the inhomogeneous sources are products of one first- and one second-order derivative, products of three first-order derivatives, and the insertion terms $e_a(\partial^2q)_b$ and $e_b(\partial^2m)_a$, along with their transposes.
\end{lemma}

\begin{proof}
Order the variables chronologically in the computational graph. For each node $u$, differentiating its update $r$ times separates the terms into a homogeneous part, linear in the $r$-th derivatives of predecessor nodes, and an inhomogeneous source
\begin{equation}\label{eq:sourceinduction}
\partial^r u=\sum_{v<u}A_{uv}\partial^r v+s_{u,e}^{(r)}.
\end{equation}
The matrices $A_{uv}$ consist of copies of $\bC,\bC^*$, bounded diagonal matrices formed from first state derivatives, and the deterministic Onsager coefficients. They depend on $\bC$ but not on $e$. Because the dependency graph has depth bounded by a constant depending on $T$, substituting \eqref{eq:sourceinduction} gives \eqref{eq:sourcepaths}, with finitely many paths and
\[
\|L_{u\nu}\|\le C_TK^{a_T}.
\]

At first order, the nonlinear chain rule is homogeneous
\[
\partial q=\partial_1f\odot\partial h,
\qquad
\partial m=\partial_1g\odot\partial b.
\]
The inhomogeneous sources arise from differentiating the matrix factors in the linear updates
\[
e_a q_b,\qquad e_b m_a.
\]
Applying the $e$-independent linearized operators gives the two source types in \eqref{eq:firstsources}.

At second order, two new source types occur. Differentiating a nonlinear update gives
\[
\partial_1^2f\odot(\partial h)^{\odot2}
\qquad\text{or}\qquad
\partial_1^2g\odot(\partial b)^{\odot2}.
\]
Substituting the first-order sources into these products generates three forms. The product of two row sources is
\[
(Pe_a)_j(\widetilde Pe_a)_j\,\alpha_b\widetilde\alpha_b
=
P_{ja}\widetilde P_{ja}\,\alpha_b\widetilde\alpha_b,
\]
which, after applying a linear operator $L$, becomes
\[
\bigl[L(P\odot\widetilde P)\bigr]_{ia}\alpha_b\widetilde\alpha_b.
\]
The product of two column sources gives the corresponding column-column form. The product of one row source and one column source gives
\[
\alpha_b\beta_a
\sum_jL_{ij}P_{ja}Q_{jb}.
\]
The second source type arises from differentiating the matrix insertion in \eqref{eq:Drecursion}:
\[
2e_a(\partial q)_b
\qquad\text{or}\qquad
2e_b(\partial m)_a.
\]
Substituting the first-order representations of $(\partial q)_b$ or $(\partial m)_a$ gives the insertion forms in \eqref{eq:secondforms}, along with their transpositions. This proves the second-order expansion.

At third order, the nonlinear chain rule in \eqref{eq:Dstructure} contributes only
\[
3\partial_1^2f\odot\partial h\odot\partial^2h,
\qquad
\partial_1^3f\odot(\partial h)^{\odot3},
\]
and the analogous terms for $g$. The linear updates contribute only
\[
3e_a(\partial^2q)_b
\qquad\text{and}\qquad
3e_b(\partial^2m)_a.
\]
Every inhomogeneous source is a finite sum of the forms
\[
D_2\bigl(\partial u_1\odot\partial^2u_2\bigr),
\qquad
D_3\bigl(\partial u_1\odot\partial u_2\odot\partial u_3\bigr),
\qquad
e_a(\partial^2q)_b,
\qquad
e_b(\partial^2m)_a,
\]
alongside their transpositions, where $D_2, D_3$ are diagonal multiplication operators bounded by the regularity constants. Because these coefficients are bounded under \eqref{eq:smoothcalc} and the number of paths is finite for fixed $T$, the source expansion and bounds hold.
\end{proof}

\begin{lemma}[First- and second-order aggregate estimates]\label{lem:aggregate12}
For the fixed matrix $\bC$, every coordinate of every iterate satisfies
\begin{align}
\sum_{a,b}|\partial_{ab}u_i|^2 &\le C_T DK^{a_T} R^2, & \sum_{a,b}|\partial_{ab}u_i|^3 &\le C_T DK^{a_T} R^3,\label{eq:aggfirst}\\
\sum_{a,b}|\partial_{ab}^2u_i|^2 &\le C_T DK^{a_T} R^4. &&\label{eq:aggsecond}
\end{align}
\end{lemma}

\begin{proof}
By Lemma~\ref{lem:finitesources}, each first-order derivative is a finite sum of the terms in \eqref{eq:firstsources}. For $p\in\{2,3\}$,
\[
\sum_{a,b}|(P e_a\alpha_b)_i|^p
=
\Big(\sum_a|P_{ia}|^p\Big)
\Big(\sum_b|\alpha_b|^p\Big)
\le
\|P\|^p N R^p
\le
D\|P\|^pR^p.
\]
The column-source term is bounded identically, replacing $N$ with $n$, giving $D\|Q\|^pR^p$. This proves \eqref{eq:aggfirst}.

At order two, Lemma~\ref{lem:finitesources} gives the forms in \eqref{eq:secondforms}. Because $P\odot\widetilde P$ is a coordinate restriction of the Kronecker product $P\otimes\widetilde P$
\[
\|P\odot\widetilde P\|
\le
\|P\otimes\widetilde P\|
=
\|P\|\,\|\widetilde P\|.
\]
The square sum over $(a,b)$ for the first form is bounded by $N R^4\|U\|^2$. For the mixed interaction, the Frobenius bound
\begin{equation}\label{eq:mixedfrobenius}
\|P^*\operatorname{diag}(L_{i,:})Q\|_F \le \|P\|\|Q\|\|L_{i,:}\|_2
\end{equation}
bounds the square sum by $R^4\|P\|^2\|Q\|^2\|L_{i,:}\|_2^2$. For the insertion forms
\begin{align*}
\sum_{a,b}|L_{ia}P_{ba}\alpha_b|^2 &\le R^2\|P\|^2\|L_{i,:}\|_2^2,\\
\sum_{a,b}|L_{ia}Q_{bb}\beta_a|^2 &\le D R^2\|Q\|^2\|L_{i,:}\|_2^2.
\end{align*}
Because $R \ge 1$, combining these terms proves \eqref{eq:aggsecond}.
\end{proof}

\subsection{The aggregate third derivative}

\begin{proposition}[Aggregate derivative of an empirical observable]\label{prop:aggregate3}
For $\phi\in C_b^{3,1}$ in \eqref{eq:observable},
\begin{equation}\label{eq:aggregate3}
\sum_{a,b}|\partial_{ab}^3H(C)|\le C_{T,\phi}DK^{a_T}R^3.
\end{equation}
Moreover, for $r=1,2,3$,
\begin{equation}\label{eq:pointH}
|\partial_{ab}^rH(C)|\le C_{T,\phi}D^{-1/2}K^{a_T}R^r.
\end{equation}
If $F=p(H)$ for a fixed polynomial $p$, the same conclusions hold with constants depending on $p$.
\end{proposition}

\begin{proof}
Decompose the trajectory vector as
\[
v_i=(u_{1i},\ldots,u_{ki},s_i),
\]
where $k$ depends on $T$, and $\partial_e s_i = 0$. For $e=(a,b)$, the third-order chain rule gives
\begin{equation}\label{eq:observablethirdfull}
\begin{split}
d\,\partial_e^3H={}
\sum_{i,\ell}\phi_\ell(v_i)\partial_e^3u_{\ell i}
+
3\sum_{i,\ell,m}\phi_{\ell m}(v_i)
\partial_eu_{\ell i}\partial_e^2u_{mi} +
\sum_{i,\ell,m,j}\phi_{\ell mj}(v_i)
\partial_eu_{\ell i}\partial_eu_{mi}\partial_eu_{ji}.
\end{split}
\end{equation}
Because the number of coordinates is fixed, sums over $\ell,m,j$ are absorbed into constants depending on $T$ and $\phi$.

For the second term in \eqref{eq:observablethirdfull}, Cauchy--Schwarz in the matrix-entry index and Lemma~\ref{lem:aggregate12} imply for every coordinate $i$
\[
\sum_e
\left|
\partial_eu_{\ell i}\partial_e^2u_{mi}
\right|\le
\left(\sum_e|\partial_eu_{\ell i}|^2\right)^{1/2}
\left(\sum_e|\partial_e^2u_{mi}|^2\right)^{1/2}\le
C_T D K^{a_T}R^3.
\]
Summing over $i$ and dividing by $d\asymp D$ gives
\[
\sum_e
\left|
\frac{1}{d}
\sum_{i,\ell,m}
\phi_{\ell m}(v_i)
\partial_eu_{\ell i}\partial_e^2u_{mi}
\right|
\le
C_{T,\phi}D K^{a_T}R^3.
\]

For the cubic first-derivative term, H\"older's inequality and \eqref{eq:aggfirst} imply
\[
\begin{aligned}
\sum_e
|\partial_eu_{\ell i}\partial_eu_{mi}\partial_eu_{ji}|
&\le
\left(\sum_e|\partial_eu_{\ell i}|^3\right)^{1/3}
\left(\sum_e|\partial_eu_{mi}|^3\right)^{1/3}
\left(\sum_e|\partial_eu_{ji}|^3\right)^{1/3}\\
&\le
C_T D K^{a_T}R^3.
\end{aligned}
\]
Summing over $i$ and dividing by $d$ preserves this bound.

For the term linear in the third trajectory derivative, fix a component $u_\ell$ and define $a_{\ell i}:=\phi_\ell(v_i)$. By Lemma~\ref{lem:finitesources},
\[
\partial_e^3u_\ell
=
\sum_\nu L_{\ell\nu}s_{\nu,e}^{(3)}.
\]
Thus
\[
\sum_i a_{\ell i}\partial_e^3u_{\ell i}
=
\sum_\nu
\left\langle
z_{\ell\nu},
s_{\nu,e}^{(3)}
\right\rangle,
\qquad
z_{\ell\nu}:=L_{\ell\nu}^*a_\ell.
\]
Because $\|a_\ell\|_2\le C_\phi\sqrt d$ and $\|L_{\ell\nu}\|\le C_TK^{a_T}$,
\begin{equation}\label{eq:dualbounds}
\|z_{\ell\nu}\|_2
\le
C_{T,\phi}\sqrt D\,K^{a_T},
\qquad
\|z_{\ell\nu}\|_1
\le
C_{T,\phi}D\,K^{a_T}.
\end{equation}

For a nonlinear source $D_2(\partial_eu\odot\partial_e^2v)$ with $\|D_2\|\le C$, applying \eqref{eq:dualbounds} and the aggregate estimate gives
\[
\begin{aligned}
\sum_e
\left|
\left\langle
z_{\ell\nu},
D_2(\partial_eu\odot\partial_e^2v)
\right\rangle
\right|
&\le
C
\sum_i|z_{\ell\nu,i}|
\sum_e|\partial_eu_i\partial_e^2v_i|\\
&\le
C_{T,\phi}
D K^{a_T}R^3
\|z_{\ell\nu}\|_1\\
&\le
C_{T,\phi}D^2K^{a_T}R^3.
\end{aligned}
\]
The powers of $K$ are absorbed into $a_T$.

For a cubic source $D_3(\partial_eu\odot\partial_ev\odot\partial_ew)$, the cubic H\"older bound implies
\[
\sum_e
\left|
\left\langle
z_{\ell\nu},
D_3(\partial_eu\odot\partial_ev\odot\partial_ew)
\right\rangle
\right|
\le
C_{T,\phi}D^2K^{a_T}R^3.
\]

For the insertion source $e_a(\partial_{ab}^2q)_b$, Lemma~\ref{lem:derivatives} implies
\[
|(\partial_{ab}^2q)_b|
\le
\|\partial_{ab}^2q\|_2
\le
C_TK^{a_T}R^2.
\]
Thus
\[
\begin{aligned}
\sum_{a,b}
\left|
z_{\ell\nu,a}(\partial_{ab}^2q)_b
\right|
&\le
C_TK^{a_T}R^2
\sum_{a,b}|z_{\ell\nu,a}|=
C_TNK^{a_T}R^2
\|z_{\ell\nu}\|_1\le
C_{T,\phi}D^2K^{a_T}R^2\\
&\le
C_{T,\phi}D^2K^{a_T}R^3,
\end{aligned}
\]
using $R\ge1$. The transpose insertion satisfies the identical bound.

Every source in the first term of \eqref{eq:observablethirdfull} has a total entry sum bounded by $C_{T,\phi}D^2K^{a_T}R^3$. Because the number of source paths is finite for fixed $T$, dividing by $d\asymp D$ and combining the terms proves
\[
\sum_{a,b}|\partial_{ab}^3H(\bC)|
\le
C_{T,\phi}DK^{a_T}R^3.
\]

For the pointwise bounds at $r=1$, the chain rule and Cauchy--Schwarz imply
\[
|\partial_eH|
\le
\frac{C_\phi}{d}
\sum_\ell
\|\partial_eu_\ell\|_1
\le
\frac{C_\phi}{\sqrt d}
\sum_\ell
\|\partial_eu_\ell\|_2
\le
C_{T,\phi}D^{-1/2}K^{a_T}R.
\]
For $r=2$, the linear term is bounded identically. The quadratic term satisfies
\[
\frac{C_\phi}{d}
\sum_i|\partial_eu_i\partial_ev_i|
\le
\frac{C_\phi}{d}
\|\partial_eu\|_2\|\partial_ev\|_2
\le
C_{T,\phi}D^{-1}K^{a_T}R^2.
\]
Because $D^{-1}\le D^{-1/2}$, this proves \eqref{eq:pointH} for $r=2$. For $r=3$, the linear term is bounded by $C D^{-1/2}K^{a_T}R^3$. The mixed term is bounded by
\[
\frac{C}{d}
\|\partial_eu\|_2\|\partial_e^2v\|_2
\le
C_{T,\phi}D^{-1}K^{a_T}R^3,
\]
and the cubic term by
\[
\frac{C}{d}
\|\partial_eu\|_\infty
\|\partial_ev\|_2
\|\partial_ew\|_2
\le
C_{T,\phi}D^{-1}K^{a_T}R^3.
\]
This proves \eqref{eq:pointH} for $r\le3$.

For $F=p(H)$ with a polynomial $p$, $\phi\in C_b^{3,1}$ implies the observable is bounded
\[
|H(\bC)|\le\|\phi\|_\infty.
\]
Thus, $p^{(j)}(H)$ is bounded for every $j$. The first and second derivatives of $p(H)$ follow from the chain rule and \eqref{eq:pointH}. At third order,
\[
\partial_e^3p(H)
=
p^{(1)}(H)\partial_e^3H
+
3p^{(2)}(H)\partial_eH\partial_e^2H
+
p^{(3)}(H)(\partial_eH)^3.
\]
The first term sum is bounded by $C DK^{a_T}R^3$. Summing over the $O(D^2)$ entries bounds the second term by
\[
C D^2
\bigl(D^{-1/2}K^{a_T}R\bigr)
\bigl(D^{-1/2}K^{a_T}R^2\bigr)
\le
C D K^{a_T}R^3,
\]
and the third by
\[
C D^2
\bigl(D^{-1/2}K^{a_T}R\bigr)^3
\le
C D^{1/2}K^{a_T}R^3
\le
C D K^{a_T}R^3.
\]
Absorbing the powers of $K$ proves the bounds for $F=p(H)$.
\end{proof}

\subsection{Moving the Taylor base point}

\begin{lemma}[Lipschitz variation of the cubic derivative]\label{lem:cubicmotion}
Let $\bC_0$ be a matrix with entry $e=(a,b)$ equal to zero. Define $\bC_y := \bC_0 + y \be_a \be_b^*$. Let $K_0 := K_{\bC_0}$ and $R_0 := \max(1, \max_{\bu,i}|u_i(\bC_0)|)$. For $F=p(H)$,
\begin{equation}\label{eq:cubicmotion}
|\partial_e^3F(\bC_y)-\partial_e^3F(\bC_0)|
\le C D^{-1/2}|y|(K_0+|y|)^{a_T}(1+|y|)^{a_T}R_0^4.
\end{equation}
The Taylor expansion satisfies
\begin{equation}\label{eq:taylorremainder}
\left|F(\bC_y)-\sum_{j=0}^3\frac{y^j}{j!}\partial_e^jF(\bC_0)\right|
\le C D^{-1/2}|y|^4(K_0+|y|)^{a_T}(1+|y|)^{a_T}R_0^4.
\end{equation}
\end{lemma}

\begin{proof}
Subtracting the linear recursions at $\bC_y$ and $\bC_0$ gives the discrepancies $y \be_a q_b(\bC_0)$ and $y \be_b m_a(\bC_0)$. Iterating the Lipschitz bounds implies
\[
\max_{\bu}\|\bu(\bC_y)-\bu(\bC_0)\|_2 \le C|y|(K_0+|y|)^{a_T}R_0.
\]
This bounds all coordinates along the segment $[0, y]$ by $C(1+|y|)(K_0+|y|)^{a_T}R_0$.

Define $\bu_{0,y}:=\bu(\bC_y)$ and $\bu_{r,y}:=\partial_e^r\bu(\bC_y)$ for $r\in\{1,2,3\}$. Let $\Delta\bu_r:=\bu_{r,y}-\bu_{r,0}$. For $r\in\{1,2,3\}$, subtracting the differentiated linear updates gives
\begin{equation}\label{eq:motionmatrix}
\begin{aligned}
\Delta \bb_r^t&=\bC_y\Delta \bq_r^t+y \be_a\be_b^* \bq_{r,0}^t
+r \be_a\be_b^*\Delta \bq_{r-1}^t-\bar\lambda_t\Delta \bm_r^{t-1},\\
\Delta \bh_r^{t+1}&=\bC_y^*\Delta \bm_r^t+y \be_b\be_a^* \bm_{r,0}^t
+r \be_b\be_a^*\Delta \bm_{r-1}^t-\bar\xi_t\Delta \bq_r^t.
\end{aligned}
\end{equation}
Define $A_{k,y}:=\partial_1^k f(\bh_y,\bx_0)$ and $\Delta A_k:=A_{k,y}-A_{k,0}$. The chain rule differences expand as
\begin{align}
\Delta \bq_1={}&A_{1,y}\odot\Delta \bh_1+\Delta A_1\odot \bh_{1,0},\label{eq:motionnonlinear1}\\
\Delta \bq_2={}&A_{1,y}\odot\Delta \bh_2+\Delta A_1\odot \bh_{2,0}
+A_{2,y}\odot(\bh_{1,y}+\bh_{1,0})\odot\Delta \bh_1
+\Delta A_2\odot \bh_{1,0}^{\odot 2},\label{eq:motionnonlinear2}\\
\Delta \bq_3={}&A_{1,y}\odot\Delta \bh_3+\Delta A_1\odot \bh_{3,0}\nonumber\\
&+3A_{2,y}\odot(\bh_{1,y}\odot\Delta \bh_2+\Delta \bh_1\odot \bh_{2,0})
+3\Delta A_2\odot \bh_{1,0}\odot \bh_{2,0}\nonumber\\
&+A_{3,y}\odot(\bh_{1,y}^{\odot 2}+\bh_{1,y}\odot \bh_{1,0}+\bh_{1,0}^{\odot 2})\odot\Delta \bh_1
+\Delta A_3\odot \bh_{1,0}^{\odot 3}.\label{eq:motionnonlinear3}
\end{align}
Analogous identities hold for $\Delta \bm_r$. Bounded higher derivatives and \eqref{eq:smoothcalc} imply that $\partial_1^k f$ are globally Lipschitz. Thus, $\|\Delta A_k\|_\infty \le C\|\Delta\bh_0\|_2$ for $k\in\{1,2,3\}$, where the subscript $0$ denotes derivative order zero. Identical bounds hold for $g$.

Induct on $r$ and $t$. A polynomial $B_y \ge 1$ in $(K_0+|y|)$ and $(1+|y|)$ bounds $\|\bu_{r,y}\|_2+\|\bu_{r,0}\|_2\le B_y R_0^r$ by Lemma~\ref{lem:derivatives}. The base case is $\|\Delta \bu_0\|_2 \le |y|B_y R_0$. Assume the bound holds for orders less than $r$. In \eqref{eq:motionmatrix}, the new terms are bounded by $|y|B_y R_0^r$. Applying the inductive hypothesis and $\|\bv\odot \bw\|_2 \le \|\bv\|_\infty\|\bw\|_2$ to \eqref{eq:motionnonlinear1}--\eqref{eq:motionnonlinear3}, the inhomogeneous terms scale as $|y|B_y R_0^{r+1}$. Iterating these bounds over $T$ steps gives
\[
\max_{\bu}\|\Delta \bu_r\|_2 \le C|y| B_y R_0^{r+1}, \qquad r \in \{0,1,2,3\}.
\]
Because initialization is independent of $\bC$, initial derivatives vanish, completing the induction.

Apply these estimates to \eqref{eq:observablethirdfull}. For a linear term, boundedness of $\phi'$ and Cauchy--Schwarz imply
\[
\frac{1}{d}\left|\sum_{i=1}^d \phi'_i\Delta u_{3,i}\right| \le \frac{1}{d}\|\phi'\|_2\|\Delta\bu_3\|_2 \le C D^{-1/2}|y|B_yR_0^4.
\]
For a product term
\[
\frac{1}{d}\left|\sum_{i=1}^d \phi_i''(\partial_eu_i)\Delta u_{2,i}\right| \le \frac{1}{d}\|\phi''\|_\infty\|\partial_e\bu\|_2\|\Delta\bu_2\|_2 \le C D^{-1}|y|\widetilde{B}_yR_0^4.
\]
For a coefficient-difference term, Lipschitz continuity of $\phi^{(3)}$ implies
\[
\left|
\phi_{\ell mj}(v_{i,y})-\phi_{\ell mj}(v_{i,0})
\right|
\le
C_\phi\|v_{i,y}-v_{i,0}\|.
\]
Because the number of coordinates is fixed,
\[
\begin{aligned}
&\frac1d
\sum_i
\left|
\phi_{\ell mj}(v_{i,y})-\phi_{\ell mj}(v_{i,0})
\right|
|\partial_eu_{\ell i,y}|
|\partial_eu_{mi,y}|
|\partial_eu_{ji,y}|\\
&\qquad\le
\frac{C}{d}
\|\Delta \bv_0\|_1
\max_{\ell}\|\partial_e\bu_{\ell,y}\|_\infty^3\\
&\qquad\le
C d^{-1/2}
\|\Delta \bv_0\|_2
\max_{\ell}\|\partial_e\bu_{\ell,y}\|_2^3\\
&\qquad\le
C D^{-1/2}|y|\widetilde B_yR_0^4.
\end{aligned}
\]
Let $\Delta\bv_0$ denote the stacked order-zero trajectory discrepancies. Terms involving $\phi'$ and $\phi''$ satisfy identical bounds without higher powers of $R_0$.

A finite exponent $a_T$ bounds all polynomial products by $C(K_0+|y|)^{a_T}(1+|y|)^{a_T}$, proving \eqref{eq:cubicmotion} for $H$. 

Because $\phi\in C_b^{3,1}$ is bounded, $|H(\bC_s)|\le\|\phi\|_\infty$. For a polynomial $p$, the derivatives $p^{(j)}$ are bounded by $C_p$. Lipschitz continuity of $H$ implies
\[
|H(\bC_y)-H(\bC_0)|
\le
\frac{C}{d}\|\Delta\bv_0\|_1
\le
C d^{-1/2}\|\Delta\bv_0\|_2
\le
C D^{-1/2}|y|B_yR_0.
\]
Thus, $|p^{(j)}(H_y)-p^{(j)}(H_0)| \le C_p|H_y-H_0|$. Subtracting the expansion
\[
\partial_e^3p(H) = p'''(H)(\partial_eH)^3 + 3p''(H)(\partial_eH)(\partial_e^2H) + p'(H)\partial_e^3H
\]
evaluated at $\bC_y$ and $\bC_0$ bounds each term by the preceding estimates. The maximum exponent is $R_0^4$, proving \eqref{eq:cubicmotion} for $F = p(H)$.

Taylor's formula is
\[
F(\bC_y) = F(\bC_0) + y\partial_e F(\bC_0) + \frac{y^2}{2}\partial_e^2 F(\bC_0) + \frac{1}{2}\int_0^y(y-s)^2\partial_e^3 F(\bC_s)\,ds.
\]
Subtracting $\frac{y^3}{6}\partial_e^3 F(\bC_0)$ and applying \eqref{eq:cubicmotion} proves \eqref{eq:taylorremainder}.
\end{proof}

\subsection{Randomized replacement}

The comparison is a Lindeberg argument \cite{Chatterjee06} with a random replacement order, as in \cite{CCKK}.

\begin{theorem}[Smooth trajectory comparison]\label{thm:swap}
Let $\bX, \bY$ be independent matrices with independent centered entries, identical variances, and satisfying \eqref{eq:workingSG} for fixed $V$. Evaluate the recursions using shared side data independent of both matrices, and identical bounded deterministic coefficients. Under \eqref{eq:smoothcalc}, for $H$ in \eqref{eq:observable}, $\phi\in C_b^{3,1}$, fixed polynomial $p$, and $0<\varepsilon<1/2$,
\begin{equation}\label{eq:swapnew}
|\E p(H(\bX))-\E p(H(\bY))|
\le C_{V,T,\phi,p,\varepsilon}n^{-1/2+\varepsilon}.
\end{equation}
\end{theorem}

\begin{proof}
For $\theta\in[0,1]$, let $\bC(\theta)$ have independent entries with marginals $\mu_{e,\theta} = \theta\mathcal L(X_e)+(1-\theta)\mathcal L(Y_e)$, independent of the side data. Let $F:=p(H)$ and let $\bC^{e,0}(\theta)$ denote $\bC(\theta)$ with entry $e=(a,b)$ set to zero.

Insertions of $X_e$ or $Y_e$ into $\bC^{e,0}(\theta)$ denote independent copies. Because $H$ is bounded, $F$ is bounded. The expectation under the product measure is a polynomial in $\theta$. Differentiating $\bigotimes_e \mu_{e,\theta}$ gives the identity
\begin{equation}\label{eq:mixtureidentity}
\frac{d}{d\theta}\E F(\bC(\theta))
=\sum_e\E\bigl\{F(\bC^{e,0}(\theta)+X_e \be_a \be_b^*)-F(\bC^{e,0}(\theta)+Y_e \be_a \be_b^*)\bigr\}.
\end{equation}

Every entry of $\bC(\theta)$ is centered, has identical variance, and satisfies the sub-Gaussian bound. Zeroing an entry preserves these bounds, so \eqref{eq:Rmoment} holds for $\bC^{e,0}(\theta)$.

Applying \eqref{eq:taylorremainder} for $y=X_e$ and $y=Y_e$ at $\bC^{e,0}(\theta)$ and taking expectations cancels the constant, linear, and quadratic terms. The cubic term is
\[
\frac16\sum_e\gamma_e\E\partial_e^3F(\bC^{e,0}(\theta)),\qquad
\gamma_e=\E X_e^3-\E Y_e^3.
\]
By \eqref{eq:workingSG}, $|\gamma_e|\le C_V n^{-3/2}$ and $\E[|X_e|^4+|Y_e|^4] \le C_V n^{-2}$. Independence, \eqref{eq:Rmoment}, and H\"older's inequality bound the polynomial expectations by $n^\varepsilon$. Because $nN = O(n^2)$ and $D^{-1/2} = O(n^{-1/2})$, the expected remainder sum is
\[
\sum_e \E |\mathrm{Rem}_e| \le C(nN)D^{-1/2}n^{-2}n^\varepsilon \le C n^{-1/2+\varepsilon}.
\]

Couple $\bC^{e,0}(\theta)$ with $\bC(\theta)$ by adding the entry. Lemma~\ref{lem:cubicmotion} bounds the expected absolute difference for a fixed $e$ by $CD^{-1/2}n^{-1/2+\varepsilon}$. Multiplying by $|\gamma_e|$ and summing over $e$ bounds the base-point difference
\[
\sum_e|\gamma_e|\,\E\left|\partial_e^3F(\bC^{e,0}(\theta))-\partial_e^3F(\bC(\theta))\right| \le C(nN) n^{-3/2} D^{-1/2}n^{-1/2+\varepsilon} \le C n^{-1/2+\varepsilon}.
\]

Proposition~\ref{prop:aggregate3} and \eqref{eq:Rmoment} bound the expected aggregate third derivative at $\bC(\theta)$:
\[
\E\sum_e|\partial_e^3F(\bC(\theta))| \le C n^{1+\varepsilon}.
\]
Multiplying by $\gamma_e$ bounds the cubic term
\[
\left|\sum_e\gamma_e\E\partial_e^3F(\bC(\theta))\right| \le C_Vn^{-3/2}\E\sum_e|\partial_e^3F(\bC(\theta))| \le C n^{-1/2+\varepsilon}.
\]

H\"older's inequality applications use intermediate exponents absorbed into $\varepsilon$. Integrating \eqref{eq:mixtureidentity} over $\theta \in [0,1]$ bounds the Taylor remainder, base-point difference, and cubic term by $O(n^{-1/2+\varepsilon})$, establishing \eqref{eq:swapnew}.
\end{proof}

\section{Proof of the Main Theorem}\label{sec:assembly}

We compare the maps at a matrix cutoff and a smoothing scale, take the dimension limit, remove the smoothing cutoffs and the observable cutoffs, and remove the matrix truncation.

\subsection{Removing the higher-order regularity}

\begin{lemma}[State-variable mollification]\label{lem:C1}
Assume $f_t,g_t$ satisfy Assumption~\ref{ass:functions}. Let $\eta\in C_c^\infty([-1,1])$ be a nonnegative mollifier of unit integral. Define
\[
f_{t,\delta}(x,s)=\int f_t(x-\delta z,s)\eta(z)\,dz,
\qquad g_{t,\delta}(x,s)=\int g_t(x-\delta z,s)\eta(z)\,dz.
\]
These mollified functions have uniformly bounded Lipschitz constants, uniform approximation error bounded by $C_T\delta$, and for each $k\ge1$:
\[
\|\partial_1^kf_{t,\delta}\|_\infty+
\|\partial_1^kg_{t,\delta}\|_\infty\le C_{T,k}\delta^{1-k}.
\]
Under the original deterministic coefficients \eqref{eq:onsagerdet}, the iterates satisfy
\begin{equation}\label{eq:mollstability}
\max_{\bu}\|\bu_\delta(\bC)-\bu(\bC)\|_{\rm av}
\le C_T\delta(1+\|\bC\|)^{a_T}.
\end{equation}
\end{lemma}

\begin{proof}
The approximation bound follows by integrating $L\delta|z|$. Convolution preserves the Lipschitz bound. Expressing differentiation as the convolution of the bounded first derivative with the $(k-1)$-th derivative of the scaled kernel gives the derivative bound, implying a Lipschitz third derivative for any fixed $\delta>0$. Applying Lemma~\ref{lem:stability} with the same matrices and coefficients establishes \eqref{eq:mollstability}. The initialization $\bq^0=f_0(\mathbf{0},\bx_0)$ has the same approximation error.
\end{proof}

\begin{lemma}[Qualitative comparison under $C^1$ regularity]\label{lem:varbound}
Under the matrix assumptions of Theorem~\ref{thm:swap} and Assumptions~\ref{ass:functions}--\ref{ass:signal}, for a test function $\phi\in C_b^{3,1}$ and $r \in \{1,2,3,4\}$,
\[
\E H(\bX)^r-\E H(\bY)^r\longrightarrow0.
\]
If $\bY=\bG$ and its auxiliary observable converges in probability to $c$, then $H(\bX)\to c$ in $L^2$.
\end{lemma}

\begin{proof}
For a fixed $\delta$, apply Theorem~\ref{thm:swap} to the smoothed maps. Theorem~\ref{thm:swap} holds for arbitrary bounded deterministic coefficients. They need not match the state evolution coefficients of the mollified functions.
The bounded gradient of $\phi$, combined with \eqref{eq:mollstability} and the uniform operator-norm moments, implies
\[
\sup_n\E|H_\delta(\bC)-H(\bC)|\le C_V\delta,\qquad \bC \in \{\bX,\bY\}.
\]
Because both observables are bounded by $\|\phi\|_\infty$, the same bound holds for their first four powers. Thus, the limit superior of each unsmoothed moment difference is bounded by $C\delta$. Taking $\delta\downarrow0$ gives the first assertion. For the final assertion, convergence in probability for the bounded observable $H(\bG)$ implies convergence of the first and second moments to $c$ and $c^2$ respectively, establishing $L^2$ convergence.
\end{proof}

\subsection{Pseudo-Lipschitz observables at fixed cutoff}

\begin{lemma}[Empirical moments]\label{lem:empmom}
For an auxiliary trajectory satisfying \eqref{eq:workingSG}, fixed empirical coordinate moments are bounded in expectation uniformly in dimension. For the true trajectory at a fixed $V$, every moment of its normalized global norm is bounded. Thus, a $\PL(2)$ empirical statistic of either trajectory is bounded in $L^2$ uniformly in dimension.
\end{lemma}

\begin{proof}
Summing \eqref{eq:coordmom} over the coordinates bounds the auxiliary sequence. The true trajectory bound follows from Lemma~\ref{lem:iteratenorms}, bounded empirical coefficients, uniform operator-norm moments, and H\"older's inequality. Because $|\phi(\bv)|\le C(1+\|\bv\|_2^2)$, the empirical average is bounded by $C(1+d^{-1}\sum_i\|\bv_i\|_2^2)$. By the fourth-moment estimate, the expectation of the square of this expression is bounded.
\end{proof}

\begin{lemma}[Cutoff and trajectory transfer]\label{lem:plextension}
Consider the auxiliary recursion \eqref{eq:auxamp} under Assumptions~\ref{ass:functions} and \ref{ass:signal}, with deterministic coefficients uniformly bounded in dimension and a matrix satisfying \eqref{eq:workingSG} at a fixed sub-Gaussian parameter $V$. If its empirical statistics converge in probability to the Gaussian state evolution limits for every bounded smooth observable, then the same convergence holds in probability and $L^1$ for every $\PL(2)$ observable.

More generally, if two trajectories have tight normalized global norms and their normalized Euclidean difference tends to zero in probability, then the difference of their $\PL(2)$ empirical statistics tends to zero in probability.
\end{lemma}

\begin{proof}
Define a smooth cutoff function $\chi_L$ equal to one on the ball of radius $L$ and zero outside the ball of radius $2L$. For any $p>2$, Lemma~\ref{lem:momprop} implies
\[
\E\frac1d\sum_i|\phi(\bv_i)-\chi_L(\bv_i)\phi(\bv_i)|
\le C L^{-(p-2)}\frac1d\sum_i\E(1+\|\bv_i\|_2^p)
\le C_{V,T,p}L^{-(p-2)}.
\]
The compactly supported continuous function $\chi_L\phi$ admits a smooth, compactly supported mollification. Applying bounded smooth convergence, removing the mollification, and taking $L\to\infty$ gives convergence. The Gaussian limit preserves these moment bounds. Lemma~\ref{lem:empmom} implies uniform integrability and $L^1$ convergence.

For the trajectory transfer, Cauchy--Schwarz applied to the $\PL(2)$ definition gives
\begin{equation}\label{eq:PLtransfer}
\left|\frac1d\sum_i\{\phi(\bv_i)-\phi(\tilde{\bv}_i)\}\right|
\le C\left(1+\frac{\|\bv\|_2}{\sqrt d}+\frac{\|\tilde{\bv}\|_2}{\sqrt d}\right)
\frac{\|\bv-\tilde{\bv}\|_2}{\sqrt d},
\end{equation}
where $\bv, \tilde{\bv}$ denote the stacked coordinate vectors. Tightness of the norms bounds the first factor, completing the proof.
\end{proof}

\subsection{Removing matrix truncation}

Assume the fourth-power uniform integrability condition of Assumption~\ref{ass:UI}. Recall that $\mu_{ij}=U_{ij}$ and
\[
\nu_4(M)=\sup_n\sup_{i,j}\E[|U_{ij}|^4;\,|U_{ij}|>M].
\]
Use the truncated and standardized matrix $\bAh_M$ defined in Section~\ref{sec:setting}, with $M$ fixed and sufficiently large. Assumption~\ref{ass:matrix} is included, since it gives $\nu_4(M)\le\nu M^{-(2\alpha-2)}\to0$.

\begin{lemma}[Uniform truncation transfer]\label{lem:truncremoval}
There exists a sequence $\epsilon(M)\to0$ such that
\begin{equation}\label{eq:optrunc}
\sup_n\E\|\bA-\bAh_M\|
\le C\{M^{-1}\nu_4(M)^{1/2}+\nu_4(M)^{1/4}\}=:\epsilon(M).
\end{equation}
Both expected operator norms are bounded uniformly in $n$ and $M$. Auxiliary trajectories with identical deterministic coefficients satisfy, for every $r>0$:
\[
\lim_{M\to\infty}\sup_n\PP\left(\max_{\bu}\|\bu(\bA)-\bu(\bAh_M)\|_{\rm av}>r\right)=0.
\]
Pseudo-Lipschitz statistics satisfy the identical approximation.
\end{lemma}

\begin{proof}
Let $Z := U\ind\{|U|>M\}$ and $s^2 := \Var(U-Z)$ for a generic entry. Centering and unit variance imply $s^2=1-\E Z^2-(\E Z)^2$. The definitions imply $\E Z^2\le M^{-2}\nu_4(M)$ and $\E Z^4\le\nu_4(M)$, so $s\ge1/2$ for large $M$. The decomposition
\[
U-U^{(M)}=(Z-\E Z)+(1-s^{-1})(U-Z-\E(U-Z))
\]
and the bound $|1-s^{-1}|\le C\E Z^2$ give $\E|U-U^{(M)}|^2\le CM^{-2}\nu_4(M)$ and $\E|U-U^{(M)}|^4\le C\nu_4(M)$. Assumption~\ref{ass:UI} implies the standardized truncated fourth moments are uniformly bounded. Theorem~\ref{thm:latala} applied to the independent centered difference entries bounds the variance by $CM^{-1}\nu_4(M)^{1/2}$ and the fourth moment by $C\nu_4(M)^{1/4}$. Theorem~\ref{thm:latala} applied to the individual matrices bounds the expected operator norms uniformly in $n$.

By Lemma~\ref{lem:stability}, the normalized trajectory discrepancy is bounded by a polynomial in the operator norms multiplied by $S\|\bA-\bAh_M\|$. This prefactor is tight in $n$ and $M$. Markov's inequality conditioned on the prefactor being bounded by $L$, followed by $M\to\infty$, proves the limit. Energy tightness and \eqref{eq:PLtransfer} imply the result for pseudo-Lipschitz statistics.
\end{proof}

\subsection{Continuous derivative tests and empirical coefficients}

\begin{lemma}[Continuous tests under small trajectory perturbations]\label{lem:continuoustest}
Assume $d^{-1}\sum_i(\|\bz_i\|_2^2+\|\tilde{\bz}_i\|_2^2)=O_\PP(1)$ and $d^{-1}\sum_i\|\bz_i-\tilde{\bz}_i\|_2^2\to0$ in probability. For any bounded continuous function $a$, the empirical average of $|a(\bz_i)-a(\tilde{\bz}_i)|$ tends to zero in probability. The pathwise equivalent holds if the hypotheses hold almost surely. Convergence for bounded Lipschitz tests together with tight empirical second moments implies convergence for all bounded continuous tests.
\end{lemma}

\begin{proof}
Let $\omega_L(\eta)$ denote the modulus of continuity of $a$ on the ball of radius $L$. The average discrepancy satisfies the bound
\[
\omega_L(\eta)+2\|a\|_\infty
\left\{\frac1{dL^2}\sum_i(\|\bz_i\|_2^2+\|\tilde{\bz}_i\|_2^2)
+\frac1{d\eta^2}\sum_i\|\bz_i-\tilde{\bz}_i\|_2^2\right\}.
\]
The result follows by invoking the tightness hypothesis to choose $L$, choosing small $\eta$, and taking $d \to \infty$. The pathwise argument is identical. For the final assertion, fix a compact ball and approximate the continuous function uniformly on that ball by a Lipschitz function. Extend this approximation to the whole space via a Lipschitz extension theorem, and truncate its range if necessary to obtain a bounded Lipschitz function on the whole space. Apply the assumed convergence to this extension and control the complement of the ball using the tight empirical second moments.
\end{proof}

\begin{proposition}[Restoration of empirical Onsager coefficients]\label{prop:onsager}
Assume the auxiliary trajectory attains the $\PL(2)$ state evolution limits, $\|\bA\|=O_\PP(1)$, and Assumptions~\ref{ass:functions}--\ref{ass:signal} hold. The empirical coefficients converge to \eqref{eq:onsagerdet}, and the maximum normalized discrepancy between the true and auxiliary trajectories tends to zero in probability.
\end{proposition}

\begin{proof}
Both recursions have tight normalized global norms via Lemma~\ref{lem:iteratenorms}, and agree on the coordinates $\bh^0,\bq^0,\bm^{-1}$. Assume inductively that the normalized discrepancies up to $\bh^t$ and $\bm^{t-1}$ vanish. Bounded, jointly continuous first state derivatives and Lemma~\ref{lem:continuoustest} imply
\[
\left|\lambda_t-\rho_n^{-1}\langle\partial_1f_t(\tilde{\bh}^t,\bx_0)\rangle\right|=o_\PP(1).
\]
The scaled auxiliary average converges to $\bar\lambda_t$ via auxiliary state evolution and the bounded-continuous-test condition. Subtracting the $\bb$-updates gives a vanishing normalized discrepancy, because $(\lambda_t-\bar\lambda_t)\tilde{\bm}^{t-1}$ is $o_\PP(1)$ times a vector of tight normalized norm. The same argument applies to $\bm^t$. Applying this argument to $\partial_1g_t$ at $\bb^t$ gives $\xi_t-\bar\xi_t=o_\PP(1)$. Subtracting the transpose update bounds $\bh^{t+1}$, completing the induction. Because $\rho_n^{-1}\to\rho^{-1}$, the scaled auxiliary average converges. The deterministic coefficient error is bounded by $\|\partial_1f_t\|_\infty|\rho_n^{-1}-\rho^{-1}|=o(1)$.
\end{proof}

\begin{proof}[Proof of Theorem~\ref{thm:main}]
Under either Assumption~\ref{ass:UI} or Assumption~\ref{ass:matrix} (which implies $\nu_4(M)\le M_{2+2\alpha}M^{-(2\alpha-2)}$), $\nu_4(M)\to0$ as $M \to \infty$. For each large fixed $M$, $\bAh_M$ satisfies \eqref{eq:workingSG} with a fixed, $M$-dependent parameter. Applying the Gaussian auxiliary limit of Theorem~\ref{thm:BM}, Lemma~\ref{lem:varbound} establishes bounded smooth auxiliary convergence. Lemma~\ref{lem:plextension} extends this to $\PL(2)$ observables. Lemma~\ref{lem:truncremoval} transfers these limits to the original auxiliary trajectory via the limits $n\to\infty$ and $M\to\infty$. 

Proposition~\ref{prop:onsager} restores the empirical coefficients, allowing \eqref{eq:PLtransfer} to transfer the trajectory statistics. By boundedness and convergence in probability, bounded statistics converge in $L^1$. Because the true statistics are uniformly bounded in $L^2$ at a fixed matrix cutoff (Lemma~\ref{lem:empmom}), the $L^1$ convergence extends to unbounded tests at that fixed cutoff. For the subsequence assertion, let $Y_n$ be the maximum absolute error in \eqref{eq:mainh} and \eqref{eq:mainb}. The convergence in probability just proved allows a deterministic choice $n_j>n_{j-1}$ with $\PP(Y_{n_j}>2^{-j})\le2^{-j}$. Borel--Cantelli gives $Y_{n_j}\to0$ almost surely. This proves Theorem~\ref{thm:main}.
\end{proof}

\section{Aggregate Variance Normalization Is Insufficient}\label{sec:sharp}

Aggregate variance normalization does not imply the conclusion of Theorem~\ref{thm:main}. The following two-row-group example satisfies $\frac1n\sum_i\E[nA_{ij}^2]=1$ for every column, but the first row multiplication $\bb^0$ has a non-Gaussian empirical mixture limit. An aggregate-normalized class does not satisfy the homogeneous state evolution conclusion.

\begin{theorem}[Aggregate normalization is insufficient]\label{thm:sharpvar}
Let Assumptions~\ref{ass:functions}--\ref{ass:signal} hold. Assume $n$ is even and $\sigma_0^2>0$. Let $\{\mu_{ij}\}$ be an array satisfying Assumption~\ref{ass:matrix}. Define the matrix $\bA$ by
\[
A_{ij}:=\sqrt{v_i}\,\mu_{ij}/\sqrt n ,\qquad
v_i:=\begin{cases} c_1, & i\le n/2,\\ c_2, & i>n/2,\end{cases}
\]
with $c_1\neq c_2$, $c_1,c_2\ge0$, and $\tfrac12(c_1+c_2)=1$. This matrix satisfies the aggregate normalization $\frac1n\sum_i\E[nA_{ij}^2]=1$, but $\bA\notin\cI$. At $t=0$, for every bounded continuous $\psi$,
\[
\frac1n\sum_{i=1}^n\psi(b^0_i)\ \xrightarrow{\ \PP\ }\ \tfrac12\,\E\,\psi\left(\sqrt{c_1}\,\sigma_0Z\right)+\tfrac12\,\E\,\psi\left(\sqrt{c_2}\,
\sigma_0Z\right).
\]
For $c_1\neq c_2$ and $\sigma_0^2>0$, there exists a bounded continuous test function, for example $\psi(x)=\cos x$, for which this limit differs from $\E\psi(\sigma_0Z)$. The conclusion of Theorem~\ref{thm:main} fails, and state evolution \eqref{eq:SE} does not extend to aggregate variance normalizations. General variance profiles require profile-dependent formulations \cite{JM13}.
\end{theorem}

\begin{proof}[Proof of Theorem~\ref{thm:sharpvar}]
The vector $\bq^0=f_0(\mathbf 0,\bx_0)$ is independent of $\bA$, and the entries $b^0_i=\sum_jA_{ij}q^0_j$ are conditionally independent given $\bq^0$. Fix a bounded Lipschitz function $\psi$. If $v_i=0$, $b_i^0=0$ and $\E[\psi(b^0_i)\mid\bq^0]=\psi(0)$. For $v_i>0$, $b^0_i=\sum_j\sqrt{v_i}\,\mu_{ij}q^0_j/\sqrt n$ is a conditional sum of independent centered terms with conditional variance $v_i\norm{\bq^0}^2/n$ and conditional Lyapunov numerator bounded by $v_i^{1+\alpha}\nu\,n^{-(1+\alpha)}\norm{\bq^0}_{2+2\alpha}^{2+2\alpha}$.

Applying Lemma~\ref{lem:slln} to $f_0(0,x_{0j})^2$ and $\abs{f_0(0,x_{0j})}^{2+2\alpha}$ implies $\norm{\bq^0}^2/N\to\E f_0(0, X_0)^2=\rho\sigma_0^2$ and $\norm{\bq^0}_{2+2\alpha}^{2+2\alpha}/N\to \E\abs{f_0(0, X_0)}^{2+2\alpha}<\infty$ almost surely. For each positive-weight row, the Lyapunov condition holds. By the one-dimensional Lindeberg comparison of Proposition~\ref{prop:lindebergt0},
\[
\E\left[\psi(b^0_i)\,\big|\,\bq^0\right]\longrightarrow\E\,\psi\left(\sqrt{
v_i}\,\sigma_0Z\right),
\]
uniformly over rows with identical $v_i$. Conditionally on $\bq^0$, the summands $\psi(b^0_i)$ are independent and bounded. The variance of their conditional average is bounded by $\norm\psi_\infty^2/n$. The conditional Chebyshev inequality and bounded convergence imply
\[
\frac1n\sum_{i=1}^n\psi(b^0_i)\xrightarrow{\ \PP\ }\tfrac12\E\psi(\sqrt{c_1}
\sigma_0Z)+\tfrac12\E\psi(\sqrt{c_2}\sigma_0Z).
\]
Set $\psi(x)=\cos(\theta x)$. For $c_1\neq c_2$ and $\theta\ne0$, the convexity of $c\mapsto e^{-c\sigma_0^2\theta^2/2}$ implies $\tfrac12e^{-c_1\sigma_0^2\theta^2/2}+\tfrac12e^{-c_2\sigma_0^2\theta^2/2} >e^{-\sigma_0^2\theta^2/2}$. The limit differs from the state evolution prediction $\E\psi(\sigma_0Z)$.

This limit extends to bounded continuous tests. Conditional weak convergence implies convergence of the row expectations, uniformly over rows with identical $v_i$ due to the uniform Lyapunov bound and tightness. If $c_j=0$, the variables are zero and the limit is degenerate.
\end{proof}

\section{Almost-Sure Convergence and the Moment Boundary}\label{sec:upgrades}

This section compares fourth centered moments, gives conditions for almost-sure transfer, and examines the moment assumptions.

\subsection{Fourth-moment concentration and almost-sure convergence}\label{sec:as}

\begin{proposition}[Fourth centered moment comparison]\label{prop:fourth}
Let $H_X,H_G$ be the auxiliary observables of Theorem~\ref{thm:swap}, with sub-Gaussian proxy $V$, maps satisfying \eqref{eq:smoothcalc}, and $\phi\in C_b^{3,1}$. For $\varepsilon>0$,
\[
\left|\E(H_X-\E H_X)^4-\E(H_G-\E H_G)^4\right|
\le C_{T,\phi,V,\varepsilon}n^{-1/2+\varepsilon}.
\]
The non-Gaussian centered fourth moment is bounded by the Gaussian centered fourth moment plus this error. Under Assumption~\ref{ass:functions} alone, difference is $o(1)$. 
\end{proposition}
\begin{proof}
Let $\mu=\E H$. The identity is
\[
\E(H-\mu)^4=\E H^4-4\mu\E H^3+6\mu^2\E H^2-3\mu^4.
\]
Theorem~\ref{thm:swap} compares raw moments up to order four. Because $|H|\le\|\phi\|_\infty$, the polynomial is Lipschitz on the bounded set of raw moments. Summing the comparison errors gives the bound. Lemma~\ref{lem:varbound} gives qualitative raw-moment comparison for $C^1$ maps. For coefficients \eqref{eq:onsagerdet}, Gaussian state evolution and boundedness imply the Gaussian centered fourth moment vanishes asymptotically. The proof compares moments without coupling $H_X$ and $H_G$.

\end{proof}

\begin{theorem}[Almost-sure subsequences and transfer]\label{thm:as}
Under Assumptions~\ref{ass:matrix}--\ref{ass:signal}:
\begin{enumerate}
\item[(a)] For maps satisfying \eqref{eq:smoothcalc}, $\phi \in C_b^{3,1}$, and a fixed truncation level, auxiliary statistics converge almost surely along any sequence $(n_j)$ satisfying $\sum_j[\E(H_{G,n_j}-\E H_{G,n_j})^4+C_{T,\phi, M,\varepsilon}n_j^{-1/2+\varepsilon}]<\infty$ for some $0<\varepsilon<1/2$. The trajectories use coefficients \eqref{eq:onsagerdet}, $H_{G,n_j}$ is the Gaussian auxiliary statistic, and the constant corresponds to Proposition~\ref{prop:fourth}.
\item[(b)] For a finite collection of pseudo-Lipschitz tests and horizons, there exists a deterministic subsequence on which all true-trajectory statistics converge almost surely. A diagonal argument extends this to countable collections.
\item[(c)] Suppose that, on a common probability space, another sequence $\mathbf C_n$ has the same dimensions and side data, its true trajectory satisfies the $\PL(2)$ limits almost surely, $\norm{\bA_n-\mathbf C_n}\to0$ almost surely, and $\norm{\bA_n}+\norm{\mathbf C_n}$ and the normalized side-data norms are bounded almost surely. Then the true trajectory of $\bA_n$ has the same almost-sure limits.
\end{enumerate}
\end{theorem}

\begin{proof}
For (a), Markov's inequality and summability imply $H-\E H\to0$ almost surely by Borel--Cantelli. Expectations converge by Theorems~\ref{thm:swap} and \ref{thm:BM}, establishing (a).

For (b), let scalar errors $Y_n^{(r)}$ converge in probability to zero (Theorem~\ref{thm:main}). Choose $n_j>n_{j-1}$ satisfying $\PP(\max_{r\le j}|Y_{n_j}^{(r)}|>2^{-j})\le2^{-j}$. Borel--Cantelli establishes simultaneous convergence for fixed $r$.

For (c), apply the perturbation estimate in Appendix~\ref{app:traj} with Lemma~\ref{lem:stability}. The continuous-derivative argument implies normalized trajectory differences vanish almost surely on the stated event. The norm bound controls empirical second moments, and the pseudo-Lipschitz Cauchy--Schwarz inequality \eqref{eq:PLtransfer} transfers each statistic.
\end{proof}

\subsection{The moment boundary}\label{sec:boundary}

Theorem~\ref{thm:main} already includes the fourth-power uniform-integrability conclusion, and Lemma~\ref{lem:truncremoval} gives the truncation step. Under Assumption~\ref{ass:matrix}, the bound $\nu_4(M)\le\nu M^{-(2\alpha-2)}$ gives, for sufficiently large $M$,
\[
\sup_n\E\|\bA-\bAh_M\|\le C\bigl(M^{-\alpha}+M^{-(\alpha-1)/2}\bigr),
\]
where $C$ may depend on the uniform moment and aspect-ratio bounds. Under Assumption~\ref{ass:UI} alone, the truncation error still tends to zero, but we do not guarantee a polynomial rate in $M$. These are truncation estimates. A rate for the full state evolution limit also requires control of the other approximation steps.

\begin{proposition}[First row update under uniform square-integrability]\label{prop:lindebergt0} 
Let $\bA$ have independent entries $A_{ij}=\mu_{ij}/\sqrt n$ with $\E\mu_{ij}=0$, $\E\mu_{ij}^2=1$, satisfying the uniform square-integrability condition
\[
\nu_2(K):=\sup_{n}\sup_{i,j}\E\left[\mu_{ij}^2\,;\,\abs{\mu_{ij}}>K\right]\
\xrightarrow[K\to\infty]{}\ 0.
\]
Under Assumptions~\ref{ass:functions}--\ref{ass:signal}, for bounded Lipschitz $\psi$, $\frac1n\sum_i\psi(b^0_i ,w_i)\to\E\psi(\sigma_0Z,W)$ in probability.
\end{proposition}

\begin{proof}
If $\sigma_0^2=0$, $f_0(0,X_0)=0$ almost surely, causing $\bq^0=\bb^0=0$ almost surely. The limit is the law of large numbers for $\psi(0,w_i)$. Assume $\sigma_0^2>0$. Condition on $\bq^0$ (independent of $\bA$). By Lemma~\ref{lem:slln} applied to $f_0(0,x_{0j})^2$ and $f_0(0,x_{0j})^4$, a.s.\ $\norm{\bq^0}^2/n\to\sigma_0^2$ and $m_\ast:=\max_j\abs{q^0_j}=O(N^{1/4})=o(\sqrt n)$. Fix row $i$, let $X_j:=A_{ij}q^0_j$ (independent, centered, variances $v_j=(q^0_j)^2/n$, $\sum_jv_j\to\sigma_0^2$, $\max_jv_j\to0$), and let $Z_j\sim\mathcal N(0,v_j)$ be independent. 

For $f\in C^3_b(\R)$, telescoping with second-order Taylor expansion gives, for $\epsilon>0$:
\[
\begin{split}
\left|\E f\left(\sum_jX_j\right)-\E f\left(\sum_jZ_j\right)\right|
\le\;&\sum_j\left(\tfrac{\norm{f^{(3)}}_\infty}6\,\epsilon\,\E X_j^2\
+2\norm{f^{(2)}}_\infty\,\E\left[X_j^2;\abs{X_j}>\epsilon\right]\right)\\
&+\sum_j\left(\tfrac{\norm{f^{(3)}}_\infty}6\,\epsilon\,v_j
+2\norm{f^{(2)}}_\infty\,\E\left[Z_j^2;\abs{Z_j}>\epsilon\right]\right).
\end{split}
\]

For $Z_j$, $\E[Z_j^2;\abs{Z_j}>\epsilon]\le \epsilon^{-2}\E Z_j^4=3\epsilon^{-2}v_j^2$. Defining $s_n:=\sum_jv_j\to\sigma_0^2$,
\[
\sum_j\E\left[Z_j^2;\abs{Z_j}>\epsilon\right]\le
3\epsilon^{-2}\sum_jv_j^2\le3\epsilon^{-2}\,s_n\max_jv_j\ \to\ 0 ,
\]
since $\max_jv_j\to0$. For $X_j$, $\abs{X_j}>\epsilon$ implies $\abs{\mu_{ij}}>\epsilon \sqrt n/\abs{q^0_j}\ge\epsilon\sqrt n/m_\ast=:K_n\to\infty$. Thus
\begin{align*}
\sum_j\E\left[X_j^2;\abs{X_j}>\epsilon\right]
\le\sum_j\frac{(q^0_j)^2}{n}\,\E\left[\mu_{ij}^2;\abs{\mu_{ij}}>K_n\right]
\le\frac{\norm{\bq^0}^2}{n}\ \nu_2\left(K_n\right)\to0 .
\end{align*}

Taking $n\to\infty$ then $\epsilon\downarrow0$ gives $\E f(\sum_jX_j)\to\E f(\sigma_0Z)$, uniformly in $i$. For bounded Lipschitz $\psi$, apply this to $f(x)=\E_W\psi(x,W)$, mollified uniformly. Independence of $W$ implies $\frac{1}{n} \sum_{i=1}^n \left|\E[\psi(b^0_i,w_i)\mid\bq^0] - \E\psi(\sigma_0Z,W)\right| \to 0$. Conditionally on $\bq^0$, bounded summands are independent across $i$. The row average has conditional variance $\le\norm\psi_\infty^2/n\to0$. Chebyshev concludes.
\end{proof}

\begin{remark}\label{rem:gapfinal}
Proposition~\ref{prop:lindebergt0} covers only $\bb^0$. The update $\bh^1$ reuses the matrix, and the proposition does not cover it. Fourth-power uniform integrability is not necessary for Gaussian state evolution, as the following array shows.

Take $n=N\ge2$, independent standard Gaussians $Z_{ij}$, Rademacher signs $\epsilon_{ij}$ and Bernoulli variables $B_{ij}$ of parameter $p_n=n^{-6}$. Define
\[
a_n=n^{3/2},\qquad c_n^2=\frac{1-n^{-3}}{1-n^{-6}},\qquad
U_{ij}=(1-B_{ij})c_n Z_{ij}+B_{ij}a_n\epsilon_{ij}.
\]
These variables are centered, independent, with variance $(1-p_n)c_n^2+p_n a_n^2=1$. Their fourth moments satisfy $3(1-p_n)c_n^4+p_n a_n^4\le4$. For every fixed $K$, choosing $n$ large enough that $a_n>K$ gives
\[
\E[U_{ij}^4;|U_{ij}|>K]\ge p_na_n^4=1.
\]
Hence $\nu_4(K)\ge1$ for every $K$, so Assumption~\ref{ass:UI} fails.

The bound $\PP(\exists i,j:B_{ij}=1)\le n^2p_n=n^{-4}$ holds. Off this event, $\bA=c_n\bG$ for $\bG=(Z_{ij}/\sqrt n)$ and $|c_n-1|=O(n^{-3})$. Because $\norm\bG=O_\PP(1)$, $\norm{\bA-\bG}=o_\PP(1)$ and both norms are bounded in probability. The perturbation estimate (Appendix~\ref{app:traj}) and pseudo-Lipschitz transfer imply Gaussian state evolution limits in probability. Fourth-power uniform integrability is not necessary for variance-normalized arrays.
\end{remark}

\section{Scope}\label{sec:scope}

AMP universality theorems differ in their hypotheses on matrix geometry, entry tails, variance profiles, side information and nonlinearities, and in their mode of convergence. The examples below isolate individual hypotheses. A theorem applies only when all its hypotheses hold, not only its moment condition.

\subsection{The comparison ensembles}

Throughout, $\varsigma$ denotes a Rademacher sign independent of all other variables, and ``standardized'' implies centered and rescaled to unit variance.

\begin{example}[Heavy tail]\label{ex:pareto}
Let $\mu_0$ have the symmetric Pareto density proportional to $\abs x^{-\beta_0-1}$ on $\abs x\ge1$ with $\beta_0=19/4$. This implies $\E\abs{\mu_0}^m=\beta_0/(\beta_0-m)$ for $m<\beta_0$ and $\E\abs{\mu_0}^m=\infty$ for $m\ge\beta_0$. Since $\E\mu_0^2=19/11$, define $\mu:=\mu_0\sqrt{11/19}$ and $A_{ij}:=\mu_{ij}/\sqrt n$ as i.i.d.\ copies. Then $\E\abs\mu^{9/2}<\infty$ and $\E\abs\mu^5=\infty$. Assumption~\ref{ass:matrix} holds with $\alpha=5/4>1$.
\end{example}

\begin{example}[All moments, no exponential tail]\label{ex:weibull}
Let $\mu_0\ge0$ have $\PP(\mu_0>x)=e^{-\sqrt x}$. Setting $x=u^2$ gives $dx=2u\,du$, and
\[
\E\mu_0^k=\int_0^\infty kx^{k-1}e^{-\sqrt x}\,dx=2k\int_0^\infty u^{2k-1}e^{-u}\,du
=2k\,\Gamma(2k)=(2k)! ,
\]
so $\E\mu_0^2=24$. Define $\mu:=\varsigma\mu_0/\sqrt{24}$ and $A_{ij}:=\mu_{ij}/\sqrt n$. Every moment is finite, satisfying Assumption~\ref{ass:matrix} for every $\alpha$. By Stirling's approximation, $(\E\abs{\mu}^k)^{1/k}\asymp k^2$. A sub-Gaussian variable has $(\E\abs X^k)^{1/k}=\mathcal{O}(k^{1/2})$ and a sub-exponential variable has $\mathcal{O}(k)$. $\mu$ is neither. Its tail does not satisfy the sub-Gaussian hypotheses in \cite[Theorems~2.1 and~2.2]{CL21} and \cite[Theorems~3.1--3.2]{Han25}.
\end{example}

\begin{example}[Two distinct entry laws]\label{ex:mixed}
Let $\mu_{ij}$ equal the standardized Pareto variable of Example~\ref{ex:pareto} for $i+j$ even, and a Rademacher sign for $i+j$ odd. Assumption~\ref{ass:matrix} holds with $\alpha=5/4$, since it constrains each entry separately. The two laws share a polynomial tail bound but have different distributions. The state evolution limit is the same as for identically distributed entries. The finite-dimensional laws of the trajectory may differ.
\end{example}

\begin{example}[Lindeberg but not $L^4$]\label{ex:lindeberg}
Let $\mu$ have the standardized symmetric Pareto law with $\beta_0=3$, implying $\E\mu^2=1$ and $\E\mu^4=\infty$. Because $\E\mu^2<\infty$ and all entries share a common law, $\nu_2(K)=\E[\mu^2;\abs\mu>K]\to0$. Proposition~\ref{prop:lindebergt0} applies at $t=0$. Assumption~\ref{ass:UI} fails, so Theorem~\ref{thm:main} does not apply to this ensemble. No conclusion for matrix-reusing updates of this ensemble is asserted here.
\end{example}

\begin{example}[Variance profile]\label{ex:profile}
Let $A_{ij}=\sqrt{S_{ij}/n}\,G_{ij}$ with independent standard Gaussians and a bounded positive deterministic profile $S$. A nonconstant profile does not imply per-entry variance normalization. Profile-dependent Gaussian AMP results appear in \cite{JM13, BHX}. The white-noise theorem \cite[Theorem~2.21]{WZF} requires the balance conditions in \cite[Definition~2.18(c)]{WZF} and the polynomial decorrelation condition in \cite[Theorem~2.21(3), equation~(2.16)]{WZF}. The non-symmetric theorem \cite{GHN} requires a square geometry and specific covariance profiles. For the two-valued row profile of Theorem~\ref{thm:sharpvar}, the first-row empirical law is a Gaussian scale mixture.
\end{example}

\begin{example}[Sparse, $n$-dependent]\label{ex:sparse}
For $n\ge2$, let $p_n=n^{-1/3}$, $\mu^{(n)}_{ij}=(B_{ij}-p_n)/\sqrt{p_n(1-p_n)}$ with $B_{ij}$ i.i.d.\ Bernoulli$(p_n)$, and $A_{ij}=\mu^{(n)}_{ij}/\sqrt n$. For every integer $p\ge3$, $n\max_{ij}\E\abs{A_{ij}}^p=\mathcal{O}(n^{(2-p)/3})\to0$, which satisfies the entry-moment condition in \cite[Definition~2.18(b)]{WZF}. However,
\[
\E(\mu^{(n)}_{ij})^4=\frac{1-3p_n(1-p_n)}{p_n(1-p_n)}\ \sim\ n^{1/3}
\longrightarrow\infty ,
\]
falling outside Assumptions~\ref{ass:matrix} and~\ref{ass:UI}. For fixed $K$, $\E[(\mu^{(n)}_{ij})^2;\abs{\mu^{(n)}_{ij}}>K]\ge\tfrac12$ for large $n$, failing the uniform square-integrability hypothesis of Proposition~\ref{prop:lindebergt0}. The classical dimension-scaled Lindeberg condition holds for bounded initialization weights because $\max_{ij}|A_{ij}|=\mathcal{O}(n^{-1/3})$. Applying \cite{WZF} requires its other hypotheses.
\end{example}

\begin{example}[Non-smooth denoiser]\label{ex:softthresh}
For $\theta>0$, the soft threshold $\eta(x;\theta)=\operatorname{sign}(x)(|x|-\theta)_+$ is Lipschitz but non-differentiable at $\pm\theta$, failing Assumption~\ref{ass:functions}. Soft thresholding satisfies the Lipschitz nonlinearity requirement in \cite[Theorem~3.2]{Han25}. Applying that theorem also requires its matrix and initialization hypotheses. In \cite[Definition~2.2 and Theorem~2.2]{CL21}, the AMP nonlinearities are $C^1$ with Lipschitz first derivatives, so soft thresholding does not satisfy those hypotheses directly. No non-smooth extension of the empirical-coefficient theorem is asserted here.
\end{example}

\subsection{The moment methods in detail}

The white-noise theorem \cite[Theorem~2.21]{WZF} requires every entry moment to be finite and satisfy the decay constraint
\[
\lim_{m,n\to\infty}\ n\,\max_{a\le m}\max_{i\le n}\E\abs{W[a,i]}^{p}=0
\qquad\text{for each integer }p\ge3.
\]
Let $m,n$ denote the dimensions in \cite{WZF}, with $m/n\to\gamma\in(0,\infty)$. The balance conditions are $\max_a\abs{\frac1n\sum_iS[a,i]-1}\to0$ and $\max_i\abs{\frac1m\sum_aS[a,i]-1}\to0$ \cite[Definition~2.18(c)]{WZF}. Assumption~\ref{ass:matrix} uses a uniform $L^{2+2\alpha}$ bound ($\alpha>1$). Theorem~\ref{thm:main} uses uniform integrability of fourth powers. Example~\ref{ex:pareto} satisfies these assumptions but violates \cite{WZF} at $p=5$. Example~\ref{ex:sparse} satisfies the moment condition of \cite{WZF} but violates the fourth-moment condition here. Both hypotheses hold for independent, variance-normalized ensembles whose moments of all orders are uniformly bounded. 

Applying \cite{WZF} requires matching its side-information, spectral, and nonlinearity assumptions. Finite side-information moments alone do not imply the polynomial-density requirement in \cite[Assumption~2.17]{WZF}. For nonconstant variance profiles, one must also check the polynomial decorrelation condition in \cite[Theorem~2.21(3), equation~(2.16)]{WZF}. Our theorem gives a fourth-power uniform integrability condition for independent entries. The $L^1$ assertion holds for bounded observables in the original ensemble and for all $\PL(2)$ observables at a fixed matrix cutoff, and the centered fourth-moment estimate retains a Gaussian remainder.

The non-symmetric theorem \cite{GHN} models correlations between transposed entries via a covariance profile. The moment-growth and sparsity conditions in \cite[Assumption~A-3]{GHN} are $(\E|X_{ij}|^k)^{1/k}\le Ck^{\nu/2}$ uniformly over entries and dimensions, and $K_n\ge C\log^{\nu\vee1}n$, where $\nu$ is the exponent used in that reference. The stretched-exponential example satisfies the moment condition with $\nu=4$, and the Pareto example fails the all-moment requirement. Our result addresses independent entries and does not assert almost-sure convergence for correlated profiles.

\subsection{Extensions and Outlook}

\begin{remark}[Smoothness of the nonlinearities]\label{rem:smoothgap}
The main theorem requires jointly $C^1$ globally Lipschitz maps. Their first state derivatives are bounded and continuous but need not be Lipschitz. Higher state derivatives are used in the regularized comparison. Lemma~\ref{lem:C1} removes this regularization while holding the deterministic coefficients fixed. Lemma~\ref{lem:continuoustest} gives the continuity for empirical Onsager coefficients.

If $\chi$ is smooth, compactly supported, and equal to one near zero, $f(x,y)=\chi(x)|x|^{3/2}$ is jointly $C^1$ and globally Lipschitz. Its first state derivative is continuous but not Lipschitz near zero, and its second derivative does not exist at zero. It satisfies the main assumption but not \eqref{eq:smoothcalc}. Soft thresholding requires separate control at the points where its derivative does not exist. 
\end{remark}

\begin{remark}[Mode of the conclusion]\label{rem:modegap}
The conclusion is an averaged asymptotic statement. Han \cite[Theorem~3.1]{Han25} proves entrywise bounds under sub-Gaussian and smoothness assumptions. A bound on the maximum over all coordinates implies bounds on any fixed set of coordinates. Our theorem gives no entrywise bounds under weaker moments. The averaged statement is all that the truncation argument needs.
\end{remark}

\begin{remark}[Dependent and inhomogeneous ensembles]\label{rem:depgap}
The theorem requires independent entries with equal variances. The variance-profile example lacks variance homogeneity, and the sparse example lacks the uniform moment hypothesis. Theorem~\ref{thm:sharpvar} demonstrates failure for an aggregate-normalized profile. The perturbation estimate in Appendix~\ref{app:traj} preserves state evolution under an $o_\PP(1)$ operator-norm perturbation when both matrix norms are tight, even if per-entry normalization fails.
\end{remark}

\appendix

\section{Proofs for the Preliminaries (Section~\ref{sec:toolbox})}\label{app:toolbox}

\begin{proof}[Proof of Lemma~\ref{lem:rosenthal}]
For $p=1$, independence implies $\E(\sum_lX_l)^{2} = \sum_l\E X_l^2$. Assume $p \ge 2$. If all second moments vanish, $X_l=0$ almost surely. Otherwise, by homogeneity of degree $2p$, normalize
\begin{equation}\label{eq:rosnorm}
\left(\sum_l a_l\right)^{p}+\sum_l b_l=1,\qquad a_l:=\E X_l^2,\quad
b_l:=\E X_l^{2p}.
\end{equation}
By the multinomial theorem, independence, and zero means, only terms with multiplicities $\ge2$ remain
\[
\E\left(\sum_lX_l\right)^{2p}
=\sum_{r=1}^{p}\ \sum_{j_1<\dots<j_r}\
\sum_{\substack{m_1+\dots+m_r=2p\\ m_i\ge2}}
\frac{(2p)!}{m_1!\cdots m_r!}\ \prod_{i=1}^r\E\left[X_{j_i}^{m_i}\right].
\]

The triangle inequality bounds each product by $\prod_{i=1}^r\E|X_{j_i}|^{m_i}$. By Lyapunov's inequality, for $2\le m\le2p$, $\E\abs{X_j}^{m} \le a_j^{\frac{2p-m}{2p-2}} b_j^{\frac{m-2}{2p-2}}$. Young's inequality $x^{1-\theta}y^{\theta}\le x+y$ implies $\E\abs{X_j}^{m} \le a_j+b_j$. Bounding the multinomial coefficient by $(2p)^{2p}$ and the compositions of $2p$ into $r$ parts by $2^{2p}$ implies
\begin{align*}
\E\left(\sum_lX_l\right)^{2p}
\ &\le\ (2p)^{2p}\,2^{2p}\sum_{r=1}^{p}\ \sum_{j_1<\dots<j_r}
\prod_{i=1}^r(a_{j_i}+b_{j_i})\\
\ &\le\ (2p)^{2p}\,2^{2p}\sum_{r=1}^{p}\frac{1}{r!}\left(\sum_j(a_j+b_j)
\right)^{r}.
\end{align*}

Because $\sum_ja_j\le1$ and $\sum_jb_j\le1$, each sum is bounded by $2^r\le2^p$. Summing over the $p$ values of $r$ bounds the total by $(2p)^{2p}2^{2p} p 2^p = 32^{p}p^{2p+1}$. For $p\ge1$, $p \le 2^{p}$, implying $32^{p}p^{2p+1}\le 64^{p}p^{2p} = (8p)^{2p} = C_p$.
\end{proof}

\begin{proof}[Proof of Lemma~\ref{lem:slln}]
Assume $y=0$. Let $M_4 := \sup_k \E Y_k^4 < \infty$. 

For (a), Lemma~\ref{lem:rosenthal} with $p=2$ and Jensen's inequality $\E Y_k^2 \le M_4^{1/2}$ imply
\[
\E\left(\sum_{k=1}^d Y_k\right)^4 \le C_2 \left[ \left(\sum_{k=1}^d \E Y_k^2\right)^2 + \sum_{k=1}^d \E Y_k^4 \right] \le 2 C_2 M_4 d^2.
\]
Markov's inequality bounds the tail probability by $\frac{2 C_2 M_4}{d^2 \epsilon^4}$. This is summable over $d \in \N$, and the Borel--Cantelli Lemma implies $\frac{1}{d}\sum_{k=1}^d Y_k \to 0$ almost surely.

For (b), row-wise independence and $M_4 = \sup_{d,k} \E Y_{d,k}^4$ give the identical row-sum bound $\E(\sum_{k=1}^d Y_{d,k})^4 \le 2 C_2 M_4 d^2$. The Borel--Cantelli Lemma requires only marginal tail probability summability, implying almost-sure convergence for the array.
\end{proof}

\subsection*{Sub-Gaussian operator-norm moments.}
If $V=0$ in \eqref{eq:workingSG}, then $\bC=0$ almost surely. Assume $V>0$, and let $\mathcal N_n$ and $\mathcal N_N$ be $1/4$-nets of the unit spheres with at most $9^n$ and $9^N$ points. For fixed unit vectors $x,y$, independence gives
\[
\E e^{z x^*\bC y}\le \exp(V^2z^2/(2n)).
\]
Since $\|\bC\|\le2\max_{x\in\mathcal N_n,y\in\mathcal N_N}|x^*\bC y|$, the exponential Markov inequality and a union bound give
\[
\PP(\|\bC\|>2r)\le2\exp\bigl((n+N)\log9-nr^2/(2V^2)\bigr).
\]
The aspect ratios are bounded away from zero. Choose a constant $r_0$ so that, for $r\ge r_0$, the right-hand side is at most $2e^{-nr^2/(4V^2)}$. Integrating this tail proves $\sup_n\E\|\bC\|^p<\infty$ for every fixed $p>0$. This is the operator-norm estimate used in Lemmas~\ref{lem:momprop} and \ref{lem:empmom}.

\section{Perturbations of the Empirical-Coefficient Recursion}\label{app:traj}

The comparison uses deterministic coefficients. The almost-sure transfer and specific examples use perturbations of the true recursion. This section details that argument, distinguishing between Lipschitz and continuous first derivatives.

\begin{proposition}[Qualitative perturbation principle]\label{prop:trueperturb}
Let $C_n,\widetilde C_n$ use identical side data and maps satisfying Assumption~\ref{ass:functions}. Assume their operator norms and normalized side-data norms are tight and $\|C_n-\widetilde C_n\|\to0$ in probability. Their true trajectories have vanishing normalized Euclidean difference at each time. If the bounds and matrix convergence hold pathwise, with eventual norm boundedness, the trajectory convergence holds pathwise.
\end{proposition}
\begin{proof}
The global energy estimate applies because uniform bounds on first state derivatives bound the empirical coefficients. The initial trajectories agree. Assume discrepancies through $h^t,m^{t-1}$ vanish. The $q$ discrepancy vanishes by Lipschitz continuity. Lemma~\ref{lem:continuoustest} applied to $(h_i^t,x_{0i})$ implies $\lambda_t-\widetilde\lambda_t\to0$ in probability. Subtraction gives
\[
\begin{split}
b^t-\widetilde b^t={}&C_n(q^t-\widetilde q^t)
+(C_n-\widetilde C_n)\widetilde q^t\\
&-\lambda_t(m^{t-1}-\widetilde m^{t-1})
-(\lambda_t-\widetilde\lambda_t)\widetilde m^{t-1}.
\end{split}
\]
Each term has vanishing normalized norm. The $m$ discrepancy follows. The continuous-test argument implies the $\xi$ discrepancy vanishes. The transpose update implies the $h$ discrepancy vanishes. The pathwise argument replaces tightness with eventual boundedness and applies the pathwise clause of Lemma~\ref{lem:continuoustest}.
\end{proof}

If the first state derivatives are globally Lipschitz, a constant times the preceding normalized trajectory difference bounds the coefficient difference, implying a polynomial Lipschitz bound in $\|C_n-\widetilde C_n\|$. Under Assumption~\ref{ass:functions}, the compact modulus of continuity need not be linear. The examples require this qualitative conclusion.

\subsection{The Gaussian input at degenerate covariances}

The Gaussian theorem holds in probability. The covariance recursion permits singular covariance matrices. The approximation proof avoids their inverses. To address degeneracies, reduce to the nonconstant-map case of \cite[Section~3.9]{BM11}.


Fix a horizon $T$. For $\epsilon\in(0,1)$, set
\begin{align*}
f_0^{\epsilon}(x,s)&=f_0(x,s)+\epsilon,\\
f_t^{\epsilon}(x,s)&=f_t(x,s)+\epsilon x\quad(1\le t\le T),\\
g_t^{\epsilon}(x,s)&=g_t(x,s)+\epsilon x\quad(0\le t\le T).
\end{align*}
For each map with the $\epsilon x$ perturbation, at most one value of $\epsilon$ makes the map constant in its first argument for almost every value of its side variable. Indeed, subtracting two such perturbed maps would make $(\epsilon_1-\epsilon_2)x$ constant in $x$. There is also at most one value for which $f_0(0, X_0)+\epsilon=0$ almost surely. Excluding these finitely many values, choose a sequence $\epsilon\downarrow0$.

Fix one admissible $\epsilon$ and, within the following induction, suppress it from the notation for the Gaussian state variables. Define
\begin{align*}
Q_0&=f_0(0,X_0)+\epsilon,\\
Q_t&=f_t^{\epsilon}(H^t,X_0)\quad(1\le t\le T),\\
M_t&=g_t^{\epsilon}(B^t,W)\quad(0\le t\le T).
\end{align*}
The state evolution covariance identities are
\[
\E[B^rB^s]=\rho^{-1}\E[Q_rQ_s],\qquad
\E[H^{r+1}H^{s+1}]=\E[M_rM_s].
\]
The Gaussian families are independent of their respective side variables. Since $\E Q_0^2>0$, the variance of $B^0$ is positive. The map $g_0^{\epsilon}(\cdot, W)$ is nonconstant on a set of side values of positive probability, so the full support of $B^0$ gives
$\E\Var(M_0\mid W)>0$. Hence $\E M_0^2>0$ and $\Var(H^1)>0$.

Suppose, for some $1\le t\le T$, that the covariance of $(H^1,\ldots,H^t)$ and the Gram matrix $(\E[Q_rQ_s])_{0\le r,s<t}$ are positive definite. Let
\[
\mathcal F^H_{t-1}=\sigma(X_0,H^1,\ldots,H^{t-1}).
\]
The conditional law of $H^t$ given $\mathcal F^H_{t-1}$ is Gaussian with strictly positive innovation variance and full support on $\R$. Continuity and nonconstancy of $f_t^{\epsilon}(\cdot,X_0)$ on a set of side values of positive probability imply
\[
c_t:=\E\bigl[\Var(Q_t\mid\mathcal F^H_{t-1})\bigr]>0.
\]
Every $Q_s$ with $s<t$ is $\mathcal F^H_{t-1}$-measurable. Thus, for every deterministic $a=(a_0,\ldots,a_{t-1})\in\R^t$,
\[
\E\left(Q_t-\sum_{s<t}a_sQ_s\right)^2\ge c_t>0.
\]
Together with the induction hypothesis, this proves positive definiteness of the enlarged $Q$-Gram matrix. The covariance of $(B^0,\ldots,B^t)$ is therefore positive definite. Applying the same argument to
$\mathcal F^B_{t-1}=\sigma(W,B^0,\ldots,B^{t-1})$ gives
\[
\E\left(M_t-\sum_{s<t}a_sM_s\right)^2
\ge \E\bigl[\Var(M_t\mid\mathcal F^B_{t-1})\bigr]>0.
\]
Thus, the enlarged $M$-Gram matrix and the covariance of $(H^1,\ldots, H^{t+1})$ are positive definite. This completes the alternating induction and verifies the nondegeneracy required by the Gaussian result at the fixed horizon. The first state derivatives remain uniformly bounded for $\epsilon\in(0,1)$, and the map differences are bounded by $\epsilon(1+|x|)$.

The energy estimate, subtraction of the matrix updates, and continuous-test argument prove convergence to the Gaussian recursion. The limit supremum over dimension of the probability that the discrepancy exceeds any positive threshold vanishes as $\epsilon \to 0$. For uniformity, restrict operator norms and normalized global norms to fixed bounds. At each time, the map perturbation forcing is at most $C\epsilon$ in normalized norm. Coefficient differences separate into a deterministic $\epsilon$ term and a continuous derivative difference. The compact-modulus proof of Lemma~\ref{lem:continuoustest} bounds the latter uniformly on the compact set. Induction over the updates proves the claim.

The Gaussian limits converge. At each time, represent the finite Gaussian vector with covariance $\Sigma_\epsilon$ as $\Sigma_\epsilon^{1/2}Z$. Continuity of the positive-semidefinite square root implies $L^2$ convergence when the covariances converge. Global Lipschitz continuity and uniform second moments imply convergence of the products defining the next covariance matrix. Induction proves covariance convergence through the horizon. Continuous bounded derivative expectations converge by this coupling and dominated convergence. Pseudo-Lipschitz expectations converge by Cauchy--Schwarz and uniform moments. This transfers the in-probability Gaussian conclusion to the original recursion. The transfer requires no uniform lower bound on covariance eigenvalues.

\subsection{An almost-sure truncation observation}

Assume the rescaled entries have uniform moments of all orders. For fixed $\kappa>0$ and large $p$,
\[
\PP(\max_{a,b}|U_{ab}|>n^\kappa)
\le C_p n^{2-p\kappa}
\]
is summable. Truncation at $n^\kappa$ eventually changes no entries almost surely. If $\sup\E|U|^p<\infty$, the removed means are bounded by $C_p n^{-\kappa(p-1)}$, and the variance shift by $C_p n^{-\kappa(p-2)}$. Off the truncation event, the Frobenius norm of the centering and standardization changes is bounded by $C_p n^{1/2-\kappa(p-3)}$. Because $p$ is arbitrary, this bound has arbitrary polynomial decay. This controls truncation. An almost-sure state evolution conclusion requires an almost-sure limit for the comparison trajectory or summable probability errors.

\section{The Comparison Estimate and Its Relation to Earlier Methods}\label{app:comparison}

\subsection{Which variable is differentiated}

In Section~\ref{sec:universality}, we take derivatives with respect to $C_{ab}$. For $C_{ab}=U_{ab}/\sqrt n$, the chain rule gives
\[
\frac{\partial^r H}{\partial U_{ab}^r}
=n^{-r/2}\frac{\partial^rH}{\partial C_{ab}^r}.
\]
The Taylor expansion evaluates matrix-entry moments scaling as $n^{-r/2}$. This produces the $n^{-3/2}$ scale for cubic coefficients and the $n^{-2}$ scale for fourth-order remainders.

Deterministic bounds do not imply convergence under arbitrary fixed replacement orders, which evaluate cubic derivatives at different intermediate matrices. Proposition~\ref{prop:aggregate3} bounds a sum at a common matrix. Identity~\eqref{eq:mixtureidentity} and Lemma~\ref{lem:cubicmotion} quantify this difference, and the displacement term has the same order as the Taylor remainder.

\subsection{First, second, and centered fourth moments}

For bounded statistics, transferring a deterministic Gaussian limit requires comparing the first two raw moments. The expansion of the fourth centered moment involves the raw moments of orders one through four.  Theorem~\ref{thm:swap} bounds arbitrary polynomials, and Proposition~\ref{prop:fourth} evaluates powers up to four.

These comparisons specify no coupling between the empirical statistics. The estimates bound distributions and moments, not random variables. Convergence proofs evaluate $\E(H-c)^2$ for a deterministic Gaussian limit $c$, rather than $\E|H(X)-H(Y)|^2$.

The finite-size comparison rate and the Gaussian state evolution rate differ. The $O(n^{-1/2+\varepsilon})$ comparison bound introduces a Gaussian finite-size remainder. Removing $C^{3,1}$ regularity via smoothing establishes qualitative convergence. It implies no polynomial rate for $C^1$ maps without a quantitative approximation modulus.

\subsection{Derivative estimates and moment methods}

Reference \cite[Sections~5--6]{CL21} applies iterate derivatives and Gaussian interpolation to symmetric matrices. The second-order source expansion here includes same-index and mixed branching terms. The Schur-product norm bound controls the former, and the Frobenius bound \eqref{eq:mixedfrobenius} controls the latter. The aggregate estimate applies to independent rectangular entries and empirical observables.

The entry-moment condition in \cite[Definition~2.18 and Theorem~2.21]{WZF} is
\[
n\max_{a,i}\E|W[a,i]|^p\longrightarrow0
\quad\text{for each integer }p\ge3.
\]
This theorem assumes variance-balance, spectral, and side-information hypotheses. Example~\ref{ex:pareto} fails this condition at $p=5$. Example~\ref{ex:sparse} satisfies it but fails the uniform fourth-moment condition of this paper. These examples compare direct matrix hypotheses, excluding unstated truncation arguments.

The model in \cite{GHN} allows correlations between transposed entries, and its moment-growth and sparsity conditions apply to that non-symmetric geometry. Non-separable iterations are treated in \cite{LWF25}.

This proof uses no external non-Gaussian universality theorems. The probabilistic inputs are Gaussian state evolution and independent-entry operator-norm bounds. The time horizon is fixed. Growing iteration counts, alternative convergence modes, and dependent matrices need separate proofs.

\section*{Acknowledgments}
This work is in part supported by the National Science Foundation (NSF) under grants CNS2451268, CNS2514415, and ITE2515378, and the Office of Naval Research (ONR) under Grant N000142112472. The authors used Grammarly to check spelling and grammar and improve the text's readability. The authors reviewed and edited the resulting text and take responsibility for all content.

\bibliographystyle{ieeetr}
\bibliography{references}

\end{document}